\documentclass[12pt]{article}
\usepackage[utf8]{inputenc}
\usepackage{comment}
\usepackage{amsmath,amssymb,amsthm,mathtools,amstext}
\usepackage{graphicx}
\usepackage{subfig}
\usepackage{subcaption}
\usepackage{tikz}
\usetikzlibrary{shapes.geometric, arrows}
\usepackage{pgfplots}
\pgfplotsset{compat=1.17}
\usepackage{pgf-pie}
\usepackage{svg}
\usepackage{float}
\newtheorem{proposition}{Proposition}[section]

\newtheorem{theorem}{Theorem}
\newtheorem{lemma}[proposition]{Lemma}

\theoremstyle{definition}

\numberwithin{equation}{section}
\usepackage{color}
\usepackage{booktabs}
\DeclareMathAlphabet{\mathpzc}{OT1}{pzc}{m}{it}

\usepackage{hyperref}
\usepackage[style=authoryear,doi=false,url=false,eprint=false,maxcitenames=1,maxbibnames=1]{biblatex}
\usepackage{lineno}

\tikzstyle{block} = [rectangle, draw, text centered, minimum width=4cm, minimum height=1cm,   text width=5cm] 
\tikzstyle{arrow} = [thick, ->, >=stealth] 
\begin{document}

\title{Estimating the growth rate of a birth and death process using data from a small sample}
\author{Carola Sophia Heinzel\footnote{Department of Mathematical Stochastics, University of Freiburg.\\
Email:
carola.heinzel@stochastik.uni-freiburg.de}\: and Jason Schweinsberg\footnote{Department of Mathematics, University of California San Diego.  Email: jschweinsberg@ucsd.edu}}
\maketitle

\begin{abstract}
The problem of estimating the growth rate of a birth and death processes based on the coalescence times of a sample of $n$ individuals has been considered by several authors (\cite{stadler2009incomplete, williams2022life, mitchell2022clonal, Johnson2023}).  This problem has applications, for example, to cancer research, when one is interested in determining the growth rate of a clone.

Recently, \cite{Johnson2023} proposed an analytical method for estimating the growth rate using the theory of coalescent point processes, which has comparable accuracy to more computationally intensive methods when the sample size $n$ is large.  We use a similar approach to obtain an estimate of the growth rate that is not based on the assumption that $n$ is large.

We demonstrate, through simulations using the R package \texttt{cloneRate}, that our estimator of the growth rate performs well in comparison with previous approaches when $n$ is small.
\end{abstract}

\section{Introduction} 

Consider a birth and death process started from one individual at time zero.  Each individual dies independently at rate $\mu$, and each individual gives birth to another individual at rate $\lambda$.  If this population does not die out quickly, then in the long run it will grow exponentially at the rate $r := \lambda - \mu$. We assume that $r > 0$, which means the birth and death process is supercritical.  Suppose that, after a large time $T$, we sample $n$ individuals uniformly at random from the population.  Our objective is to estimate the growth rate $r$ from the genealogical tree of these $n$ sampled individuals.  We assume that mutations are neutral, so that the birth and death rates are the same for all individuals.

In applications, there are various methods available for estimating the genealogical tree from whole-genome single-cell DNA data at a single time point (e.g. \cite{de2013phylogenetic, kang2022sieve, kozlov2022cellphy, drummond2007beast, schliep2011phangorn, zafar2017sifit}). In the sequel, we assume that the genealogical tree, including the coalescence times when ancestral lines merge, can be recovered exactly.
Note that if the birth, death, and mutation rates are multiplied by $c$, while $T$ is multiplied by $1/c$, the expected pattern of mutations observed in the sample would be unchanged.  Therefore, reconstruction of the tree requires additional information that can be used to set the time axis.  This could come from knowledge of the mutation rate or, in the case of clinical data such as that analyzed by \cite{Johnson2023}, the patient's age.

The problem of estimating the growth rate is of particular interest in the study of cancer, where the birth and death process models a growing population of cancer cells and one wants to estimate the rate at which a clone is expanding. The growth rate is an important characteristic of cancer which can be interpreted as a measure of fitness (\cite{brown2023updating}).  Understanding the growth rate may help determine appropriate intervals for cancer screening (\cite{peer1993age}), and may inform treatment decisions and contribute to improved cancer prevention strategies (\cite{Johnson2023}). Beyond oncology, the concept of growth rate also plays an important role in epidemiology, where it reflects the speed at which a virus spreads through a population (\cite{stadler2011inferring, stadler2013uncovering}).

Earlier work on estimating the growth rate used computationally intensive methods. \cite{stadler2009incomplete} developed an approach based on Markov Chain Monte Carlo (MCMC) for estimating the growth rate of a birth and death process when each individual alive at time $T$ is sampled with some fixed probability.  \cite{williams2022life} and \cite{mitchell2022clonal} used an approach called Phylofit to estimate the growth rate.
Phylofit is based on computing the joint likelihood of the coalescence times, i.e. the times when the ancestral lineages of the sampled cells merge, for a model in which the population size levels off after an initial period of exponential growth. 

Recent theoretical developments have made it possible to estimate the growth rate by analytical methods.  \cite{Durrett2013} computed the site frequency spectrum for a supercritical birth and death process, that is, the expected number of mutations appearing on $k$ out of $n$ individuals sampled at a large time $T$. \cite{harris2020} and \cite{LAMBERT201830} developed a method to describe the exact genealogy of a sample of size $n$ from a birth and death process at time $T$.  Using these results, \cite{Johnson2023} proposed two closely related methods of obtaining point and interval estimates for the growth rate $r$ without using MCMC.  One method, which is entirely analytical, uses the lengths of the internal branches which can be calculated from the coalescence times in the coalescent tree  to estimate the growth rate, while the other method involves maximum likelihood estimation. 

\cite{Johnson2023} demonstrate that their methods estimate the growth rate well for large sample sizes, providing comparable accuracy to previous methods (\cite{stadler2009incomplete, williams2022life, mitchell2022clonal}) at much less computational cost. They apply their methods to blood cancer data, estimating the growth rate in 42 different clones and obtaining similar results to those obtained with more computationally intensive methods.

Nevertheless, because the methods proposed by \cite{Johnson2023} are derived using asymptotics for large $n$, they do not perform as well as the more computationally intensive methods in small samples.  Also, the estimator based on internal branch lengths depends not only on the coalescence times but also on the shape of the genealogical tree.  This adds variability to the estimate because the tree shape contains no information about the growth rate $r$.  Indeed, regardless of the growth rate, when one traces back the genealogy of $n$ sampled lineages in a birth and death process, at any time each remaining pair of lineages is equally likely to be the next pair that merges.

In this paper, we propose a modification of the estimators in \cite{Johnson2023} which is not built on the assumption that $n$ is large.  We avoid using information about the tree shape, which reduces the variance of the estimator.
We compare the new approach to previous approaches using both simulated data and real data.  We also show analytically that our estimator has, even for large $n$, less variability than the estimator based on internal branch lengths proposed by \cite{Johnson2023}.

\section{Methods for estimating the growth rate}

\subsection{Internal and external branch lengths}

Consider the genealogical tree of a sample of $n$ individuals from a birth and death process at time $T$.  Let $L_{k,n}$ denote the sum of the lengths of all branches that are ancestral to $k$ of the $n$ sampled individuals.  \cite{Durrett2013} showed that as $T \rightarrow \infty$, we have $$\mathbb E[L_{1,n}] \rightarrow \infty, \qquad \mathbb E[L_{k,n}] \rightarrow \frac{n}{rk(k-1)} \quad\mbox{for }k = 2, 3, \dots, n-1.$$
Branches that are ancestral to two or more sampled individuals are said to be internal, while branches that are ancestral to only one of the sampled individuals are said to be external (see Figure \ref{branch} for a visualization).  Therefore, if we denote by $L^{in}$ the total length of the internal branches, then Durrett's results imply that as $T \rightarrow \infty$, we have $$\mathbb E[L^{in}] \rightarrow \frac{n}{r} \sum_{k=2}^{n-1} \frac{1}{k(k-1)} = \frac{n}{r} \bigg(1 - \frac{1}{n-1}\bigg).$$

Note that this expression is inversely proportional to $r$, which implies that long internal branch lengths indicate a low growth rate, whereas short internal branch lengths indicate a higher growth rate.  Based on this result, \cite{Johnson2023} proposed the estimator
\begin{equation}\label{ILestimator}
\hat r_{Lengths} := \frac{n}{L^{in}}.
\end{equation}

\cite{Johnson2023} also derived further theoretical properties of this estimator, which they used to obtain asymptotically valid confidence intervals for $r$. \cite{schweinsberg2025asymptotics} then used similar methods to obtain second-order asymptotics for the full site frequency spectrum. 

\begin{figure}[h!]
\centering
\begin{tikzpicture}[scale = 1.2]

\draw(0, 0) -- (6, 0);

\draw[orange, thick] (0, 0) -- (0, 3.7); 
\draw[purple, thick] (0, 3.2) -- (0, 3.7); 

\draw[thick] (0, 3.7) -- (0,4);
\draw[blue, thick] (0, 0.0) -- (0, 2.5);

\draw[green!70!black, thick] (0.5, 1.8) -- (0.5, 2.5);
\draw[dashed] (0.5, 2.5) -- (0, 2.5);
\node at (0.5, 2.5) [right] {$H_{1}$};
\draw[blue, thick] (0.5, 0.0) -- (0.5, 1.8);

\draw[dashed] (1.5, 1.8) -- (0.5, 1.8);
\draw[blue, thick] (1.5, 0.0) -- (1.5, 1.8);
\node at (1.5, 1.8) [right] {$H_{2}$};

\draw[green!70!black, thick] (2.5, 2.0) -- (2.5, 3.2);
\draw[dashed] (2.5, 3.2) -- (0, 3.2);
\node at (2.5, 3.2) [right] {$H_{3}$};
\draw[blue, thick] (2.5, 0.0) -- (2.5, 2.0);

\draw[dashed] (3.0, 2.0) -- (2.5, 2.0);
\node at (3.0, 2.0) [right] {$H_{4}$};
\draw[blue, thick] (3.0, 0.0) -- (3.0, 2.0);

\draw[green!70!black,thick] (3.8, 0) -- (3.8, 3.0);
\draw[yellow!85!orange,thick] (3.8, 3.0) -- (3.8, 3.7);

\draw[dashed] (3.8, 3.7) -- (0, 3.7);
\node at (3.8, 3.7) [right] {$H_{5}$};
\draw[blue, thick] (3.8, 0.0) -- (3.8, 2.8);

\draw[dashed] (4.5, 2.8) -- (3.8, 2.8);
\node at (4.5, 2.8) [right] {$H_{6}$};
\draw[blue, thick] (4.5, 0.0) -- (4.5, 2.8);

\draw[dashed] (5.2, 3.0) -- (3.8, 3.0);
\node at (5.2, 3.0) [right] {$H_{7}$};
\draw[green!70!black, thick] (5.2, 1.5) -- (5.2, 3.0);
\draw[blue, thick] (5.2, 0.0) -- (5.2, 1.5);

\draw[dashed] (5.8, 1.5) -- (5.2, 1.5);
\node at (5.8, 1.5) [right] {$H_{8}$};
\draw[blue, thick] (5.8, 0.0) -- (5.8, 1.5);

\node at (-0.25,0.05) {$T$};
\node at (-0.25, 4) {$0$};
\begin{scope}[shift={(7.2,2.7)}]
  \draw[thick, black] (0,1.5) -- (0.6,1.5) node[right, black] {$L_{9,9}$};
  \draw[thick, purple] (0,1.2) -- (0.6,1.2) node[right, black] {$L_{5,9}$};
  \draw[thick, yellow!75!orange] (0,0.9) -- (0.6,0.9) node[right, black] {$L_{4,9}$};
  \draw[thick, orange] (0,0.6) -- (0.6,0.6) node[right, black] {$L_{3,9}$};
  \draw[thick, green!70!black] (0,0.3) -- (0.6,0.3) node[right, black] {$L_{2,9}$};
  \draw[thick, blue] (0,0) -- (0.6,0) node[right, black] {$L_{1,9}$};
\end{scope}

\end{tikzpicture}

\caption{Genealogy of a sample of size $n = 9$ from a birth and death process. The green lines represent internal branches. The blue lines represent external branches.  The 8 coalescence times are indicated by $H_1, \dots, H_8$.}\label{branch}
\end{figure}
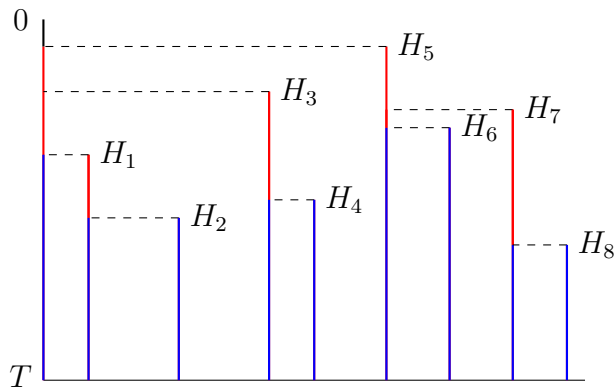

\subsection{Exact and approximate distributions of the coalescence times}

For a genealogical tree of a sample of size $n$ from a birth and death process, there are $n-1$ coalescence times.  \cite{harris2020} and \cite{LAMBERT201830} showed that there is a relatively simple way to represent these coalescence times for the genealogy of a sample of size $n$ from a supercritical birth and death process at time $T$.

To do this, we define
\begin{align*}
\delta_T := 
\frac{re^{-rT}}{\lambda (1-e^{-rT}) + re^{-rT}}.
\end{align*}
We can then generate the distribution of the coalescence times $(H_{i,n,T})_{i=1}^{n-1}$ in the following way:
\begin{itemize}
    \item[1.] Define a random variable $Y_{n,T}$ with density on $(0,1)$ given by $$f_{Y_{n,T}}(y) := \frac{n \delta_T y^{n-1}}{(y + \delta_T - y \delta_T)^{n+1}}.$$
   \item[2.] Conditional on $Y_{n,T} = y$, let the random variables $H_{1,n,T}, \dots, H_{n-1,n,T}$ be i.i.d. with density on $(0,T)$ given by
   \begin{equation}\label{HinT}
   f_{H_{i,n,T}|Y_{n,T}=y}(t) := \frac{y \lambda + (r - y \lambda) e^{-rT}}{y \lambda (1 - e^{-rT})} \cdot \frac{y \lambda r^2 e^{-rT}}{(y \lambda + (r - y \lambda)e^{-rt})^2}.
   \end{equation}
\end{itemize}
Given these coalescence times, one can generate the genealogical tree in two equivalent ways.  One method, which is implicit in the work of \cite{harris2020}, is to work backwards starting with $n$ lineages at time $T$ and merge two lineages, chosen uniformly at random, at each coalescence time.  Another method, proposed by \cite{LAMBERT201830}, is to use the coalescent point process, an approach which originated in \cite{popovic2004} and was developed further by \cite{Aldous_Popovic_2005}, \cite{stadler2009incomplete}, and \cite{ lambert2013birth}.  We draw a vertical line of height $T$, followed by vertical lines of heights $H_{1,n,T}, \dots, H_{n-1,n,T}$.  At the top of each vertical line, we draw a horizontal dashed line to the left, stopping when we hit another vertical line.  This procedure for generating a coalescent tree is illustrated in Figure~\ref{branch}.

\cite{Johnson2023} derived two successive approximations to these coalescence times, first taking a limit as $T \rightarrow \infty$ while $n$ is fixed, and then taking a limit as $n \rightarrow \infty$.  The first of these limits yields the following way to approximate the coalescence times.
\begin{itemize}
    \item[1.] Choose $Q_{n}$ from the density
    $$
    f_{Q_{n}}(q) := 
    \frac{nq^{n-1}}{(1+q)^{n+1}}, \quad q \in (0,\infty).$$
   \item[2.] Choose $(U_{i,n})_{i=1,.., n-1}$ independently from the conditional density
   \begin{equation}\label{Undensity}
   f_{U_{i,n}|Q_{n}=q}(u):=  
   \frac{1+q}{q} \cdot \frac{e^u}{(1+e^u)^2}, \quad u \in (-\log q, \infty).
   \end{equation}
   \item[3.] Let $H_{i,n} := T - \frac{1}{r}(\log Q_{n} + U_{i,n}).$
\end{itemize}
After taking the limit as $T \rightarrow \infty$, we can take the limit as $n \rightarrow \infty$.  Then the distribution of the coalescent times converges to a distribution that can be represented in the following way:
\begin{itemize}
    \item[1.] Let $W$ have an exponential distribution with mean $1$.
   \item[2.] Choose $(U_{i})_{i=1,.., n-1}$ independently from the logistic density
   \begin{equation}\label{Udensity}
   f_{U_i}(u):= \frac{e^u}{(1+e^u)^2}, \quad u \in (-\infty, \infty).
   \end{equation}
   \item[3.] Let \begin{align}
       H_i := T - \frac{1}{r}(\log(1/W) + \log n + U_i).
       \label{eq:H_i}
   \end{align}
\end{itemize}

While the work of \cite{Johnson2023} focused on the second of these limits, we will consider primarily the first of these limits, in which $n$ is fixed and $T \rightarrow \infty$, i.e. $(H_{i,n})_{i=1,..., n-1}$.  This approach has the benefit that the resulting estimates work well for small sample sizes, while the distribution of $H_{i,n}$ remains mathematically tractable.  

Based on \eqref{eq:H_i}, \cite{Johnson2023} also proposed a maximum likelihood estimator (MLE) to estimate $r.$  They wrote $H_i = a + b U_i$ with $b = 1/r$ and $a = T - \frac{1}{r}\left(\log(1/W) + \log(n) \right).$ Since $U_i$ is symmetric, this is equivalent to setting $H_i = a - b U_i.$ Then they used the R-package \texttt{maxLik} (\cite{henningsen2011maxlik}) to calculate the MLE for $b.$ The reciprocal of this estimator is denoted by $\hat r_{MLE}.$

\subsection{A modified estimator for the growth rate}

We will consider here the second of the two approximations discussed in the previous section.
Recall that the estimator $\hat r_{Lengths}$ in \eqref{ILestimator} was based on the internal branch length.  Referring to Figure \ref{branch}, the genealogical tree obtained using the coalescent point process has a branch on the left of height $T$, which we call the $0$th branch, and branches of heights $H_{1,n}, \dots, H_{n-1,n}$.  For $i = 1, \dots, n-2$, the length of the internal portion of the $i$th branch is $(H_{i,n} - H_{i+1,n})^+$, where $x^+ = \max\{x, 0\}$ for a real number $x$. The $(n-1)$st branch is entirely external, and the internal portion of the $0$th branch is $\max_{1 \leq i \leq n-1} H_{i,n} - H_{1,n}$.  Therefore, the total internal branch length is \begin{align}
    L_n^{in} = \bigg( \max_{1 \leq i \leq n-1} H_{i,n} - H_{1,n} \bigg) + \sum_{i=1}^{n-2} (H_{i,n} - H_{i+1,n})^+. \label{eq:IBL}
\end{align}

Note, however, that $L^{in}$ depends on the order in which the coalescence times $H_{1,n}, \dots, H_{n-1,n}$ are listed, even though, as can be seen from the construction in \cite{harris2020}, the order of the coalescence times provides no information about the growth rate $r$.  More formally, if we denote by $H_n^{(1)} \leq \dots \leq H_n^{(n-1)}$ the order statistics of the coalescence times, then the vector $(H_n^{(1)}, \dots, H_n^{(n-1)})$ is a sufficient statistic for estimating $r$, as all $(n-1)!$ orders of the coalescence times are equally likely regardless of the value of $r$.  Therefore, by the Rao-Blackwell Theorem, it should be possible to obtain an estimator with lower or equal variance by conditioning on the order statistics.

When we average over all possible orderings, $H_{1,n}$ in the first term gets replaced by the average of the $H_{i,n}$.  Each of the $n-2$ terms in the sum has a mean equal to the average of $(H_{i,n} - H_{j,n})^+$ over the $(n-1)(n-2)$ possible choices for $i$ and $j$.  Therefore,
    \begin{align}
        \mathbb E\left[L_n^{in}|H_n^{(1)}, \ldots, H_n^{(n-1)}\right] &= \bigg( H^{(n-1)}_n - \frac{1}{n-1} \sum_{i = 1}^{n-1} H_{i,n} \bigg)  + \frac{1}{n-1}\sum_{i = 1}^{n-1} \sum_{j = 1}^{n-1}(H_{i,n} - H_{j,n})^+ \label{eq:ev} \\
        &= \bigg( H^{(n-1)}_n - \frac{1}{n-1} \sum_{i = 1}^{n-1} H_{i,n} \bigg)  + \frac{1}{n-1} \sum_{i=1}^{n-1} \sum_{j=i+1}^{n-1} |H_{i,n} - H_{j,n}|. \notag
    \end{align}
A detailed derivation of \eqref{eq:ev} can be found in Appendix \ref{app:der}.
The second term on the right-hand side of \eqref{eq:ev} dominates when $n$ is large, which suggests that one can obtain a good estimator of $r$ that is based on the average of the terms $(H_{i,n} - H_{j,n})^+$.  Note that including the first term would double the weight assigned to the terms $(H_{i,n} - H_{j,n})^+$ in which $H_{i,n}$ is the largest of $H_{1,n}, \dots, H_{n-1,n}$, but we prefer to pursue the simpler approach of treating all of the terms $H_{i,n} - H_{j,n}$ symmetrically.  Therefore, we consider estimators of $r$ which are inversely proportional to the double sum in \eqref{eq:ev}.  
That is, we consider estimators of the form
\begin{equation}\label{rform}
    \hat r =  \frac{c(n) (n-1)(n-2)}{\sum_{i=1}^{n-1}\sum_{j=1}^{n-1}(H_{i,n} - H_{j,n})^+} ,
\end{equation}
where $c(n)$ is a constant depending only on $n$.  There are $(n-1)(n-2)$ terms in the denominator with $i \neq j$, and $\hat r$ equals $c(n)$ times the reciprocal of the average of these terms.  Note that when this estimator is applied to real data, the coalescence times derived from DNA sequence data will be used in place of $H_{i,n}$ and $H_{j,n}$.

\subsection{The choice of \texorpdfstring{$c(n)$}{c(n)}}

The best method for choosing $c(n)$ may depend on the nature of the application.  We propose here three methods:
\begin{enumerate}
\item Let $c_{MSE}(n)$ be the value of $c(n)$ that minimizes the mean squared error of $\hat r$ as an estimator of $r$.

\item Let $c_{Bias}(n)$ be the value of $c(n)$ that makes $\hat r$ an unbiased estimator of $r$.

\item Let $c_{Inv}(n)$ be the value of $c(n)$ that makes $1/\hat r$ an unbiased estimator of $1/r$.
\end{enumerate}
We denote the corresponding estimators by $\hat r_{MSE}$, $\hat r_{Bias}$, and $\hat r_{Inv}$.

Because $H_{i,n} = T - \frac{1}{r}(\log Q_n + U_{i,n})$, we have $H_{i,n} - H_{j,n} = \frac{1}{r}(U_{j,n} - U_{i,n})$, where the random variables $(U_{i,n})_{i = 1, \dots, n-1}$ are conditionally i.i.d. given $Q_n$ with density given in \eqref{Undensity}.  Therefore, $$\hat r = r c(n) S_n, \qquad \qquad S_n = \frac{(n-1)(n-2)}{\sum_{i=1}^{n-1} \sum_{j=1}^{n-1} (U_{i,n} - U_{j,n})^+}.$$
The mean squared error is given by $$\mathbb E[(\hat r - r)^2] = r^2 \Big( 1 - 2c(n) \mathbb E[S_n] + c(n)^2 \mathbb E[S_n^2] \Big),$$
and minimizing this expression gives $$c_{MSE}(n) = \frac{\mathbb E[S_n]}{\mathbb E[S_n^2]}.$$
Also, we have $$c_{Bias}(n) = \frac{1}{\mathbb E[S_n]}$$
and $$c_{Inv}(n) = \mathbb E \bigg[ \frac{1}{S_n} \bigg].$$
Because these constants are determined by the distribution of $S_n$, they depend only on $n$ and not on $r$.  Also, because it is easy to simulate a large number of random variables having the distribution of the $U_{i,n}$, we can determine the values of $c_{MSE}(n)$ and $c_{Bias}(n)$ by simulation.

It was shown in \cite{Johnson2023} that $\mathbb E[(U_i - U_j)^+] = 1$, where $U_i$ and $U_j$ are independent with density given by \eqref{Udensity}.  Therefore, we would expect $c_{MSE}(n)$, $c_{Bias}(n)$, and $c_{Inv}(n)$ all to tend to $1$ as $n \rightarrow \infty$.  The exact value of $c_{Inv}(n)$ can be calculated analytically.  We will show in the Appendix that
\begin{align}
    c_{Inv}(n) = \frac{n}{n-2} \bigg(1 - \frac{1}{n-1} \sum_{k=1}^{n-1} \frac{1}{k} \bigg). \label{eq:inv}
\end{align}

\subsection{Confidence Intervals}\label{sec:CI}

As in previous work (\cite{stadler2009incomplete}, \cite{williams2022life}, \cite{mitchell2022clonal}, \cite{Johnson2023}), we can obtain a confidence interval for the growth rate $r$ in addition to obtaining point estimates.  Through simulation, we can find the $0.025$ and $0.975$ quantiles of the distribution of $S_n$, denoted respectively by $q_{0.025}$ and $q_{0.975}$, such that $\mathbb P(q_{0.025} < S_n < q_{0.975}) = 0.95$.
Then, if we take $\hat r$ to be the estimator with $c(n) = 1$, so that $\hat r = r S_n$, we get
$$\mathbb P \bigg(\frac{\hat r}{q_{0.975}} < r < \frac{\hat r}{q_{0.025}} \bigg) = 0.95.$$
That is, we can take $\hat r/q_{0.975}$ and $\hat r/q_{0.025}$ to be the lower and upper endpoints respectively of a 95$\%$ confidence interval for $r$.

\section{Evaluation of the estimators}

\subsection{Simulation Results}

We conducted a simulation study to evaluate our estimators $\hat r_{MSE}$, $\hat r_{Inv}$, and $\hat r_{Bias}$, using simulation scenarios similar to those in \cite{Johnson2023}. Specifically, we simulated coalescent point processes for different values of $n$, $r$, and $T$,
repeating each scenario 1{,}000 times for each sample size and growth rate. The simulations were carried out using the R package \texttt{cloneRate}. Following \cite{Johnson2023}, we drew $\lambda$ from the uniform distribution on $[r, 1+r]$ and subsequently set $\mu = \lambda - r.$
We compared the performance of our estimators to the performance of the estimators $\hat{r}_{MCMC}$, $\hat{r}_{Lengths}$, and $\hat{r}_{MLE}$.  Here $\hat{r}_{MCMC}$ is the Markov chain Monte Carlo estimator based on the work of \cite{stadler2009incomplete}, which is implemented in \texttt{cloneRate}. Since this is a Bayesian method, its performance will depend on the prior distribution. \cite{williams2022life} proposed uniform priors, and we chose the default prior for $r$ of a uniform distribution on $[0,4]$.

\subsubsection{Densities of the estimators}

The densities of the six estimators are shown in Figure \ref{fig:d1} for
$r = 0.5$, $T = 40$, and $n = 5$
and in Figure \ref{fig:d4} for $r = 0.5$, $T = 40$, and $n = 20$.
These plots show that even for $n = 5$, the modes of the densities of $\hat r_{Lengths}$, $\hat r_{MLE}$, $\hat r_{MCMC}$, and $\hat r_{Inv}$ are similar to the true growth rate.  According to the density plots, $\hat r_{Bias}$ and especially $\hat r_{MSE}$ typically underestimate the growth rate for small $n$.  Moreover, the densities are all very skewed for $n = 5$. 

For $n = 20,$ the differences among the estimators are much smaller than for $n = 5.$  The skewness decreases with increasing sample size so that for $n = 20$, the densities look similar to the density of a normal distribution. This aligns with results in \cite{Johnson2023}.

To understand Figure \ref{fig:d1}, observe that if all of the coalescence times $H_{i,n}$ coincidentally happen to be very close to one another, the denominator in \eqref{rform} will be small, making the estimate $\hat r$ very large.  Although this is exceptionally rare when $n = 20$, and quite rare even when $n = 5$, it happens often enough when $n = 5$ to have a substantial influence on the mean squared error, and to a lesser extent the bias, of the estimate.  Therefore, a small value of $c(n)$, which mitigates the effect of these rare events, will minimize the mean squared error. Consequently, the estimator $\hat r_{MSE}$ that minimizes the mean squared error often underestimates $r$ by nearly a factor of $2$ when $n$ is small.

\begin{figure}[h!]
   \begin{minipage}{0.48\textwidth}
        \centering
        \includegraphics[width=\linewidth]{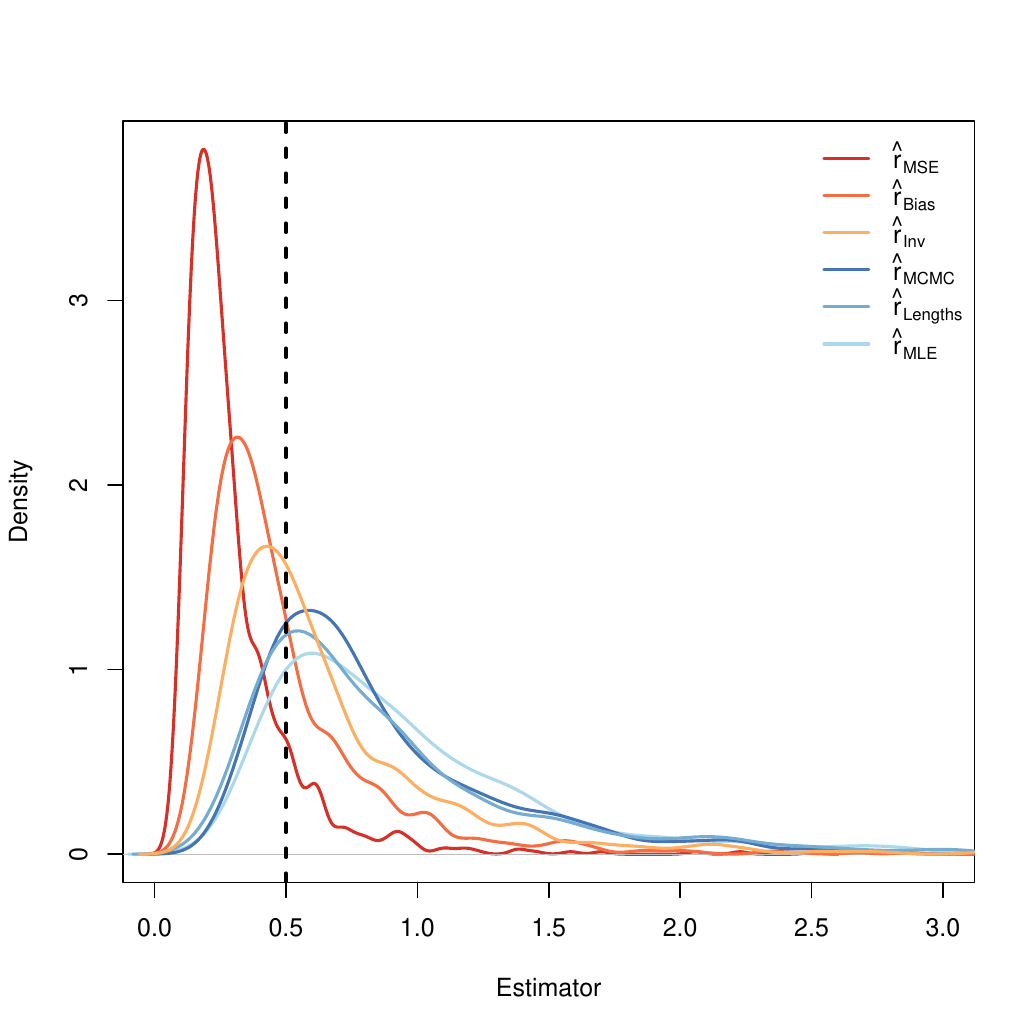}
        \caption{Densities of the six estimators when $r = 0.5$, $T = 40$, and $n = 5$. The black dotted line indicates the true growth rate.}
        \label{fig:d1}
    \end{minipage}
    \hfill  
    \begin{minipage}{0.48\textwidth}
        \centering
        \includegraphics[width=\linewidth]{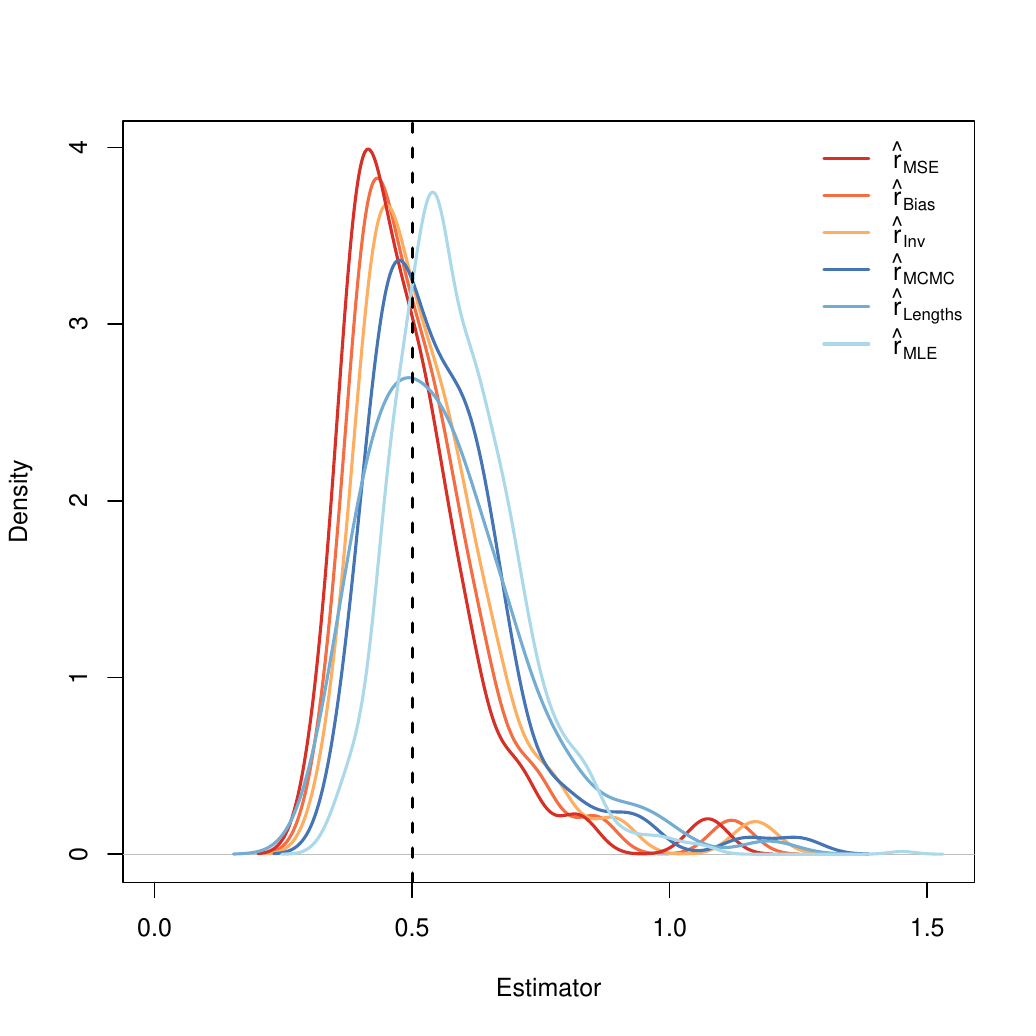}
    \caption{Densities of the six estimators when $r = 0.5$, $T = 40$, and $n = 20$. The black dotted line indicates the true growth rate.}
        \label{fig:d4}
    \end{minipage}
\end{figure}

\subsubsection{Mean Squared Error, Mean Absolute Error, and Bias}

We compared the root mean squared error (RMSE), the mean absolute error, and the bias for the six different estimators when $r = 0.5$ and $T = 40$, and when $r = 1$ and $T = 40$, for several different values of $n$.
The root mean squared errors are shown in Figures~\ref{fig:mse1} and~\ref{fig:mse2}. These figures show that the RMSE of the estimator \( \hat{r}_{{MSE}} \) is the lowest across all sample sizes. The differences among the estimators decrease as \( n \) increases.

Moreover, not only does \( \hat{r}_{{MSE}} \) outperform \( \hat{r}_{{MCMC}}, \hat{r}_{{MLE}} \) and \( \hat{r}_{{Lengths}} \) on this criterion, but also \( \hat{r}_{{Bias}} \) and \( \hat{r}_{{Inv}} \) achieve lower RMSEs than these previously existing methods when $n \geq 7$ in the case $r = 0.5$ and when $n \geq 9$ in the case $r = 1$.

\begin{figure}[h!]
   \begin{minipage}{0.44\textwidth}
        \centering
        \includegraphics[width=\linewidth]{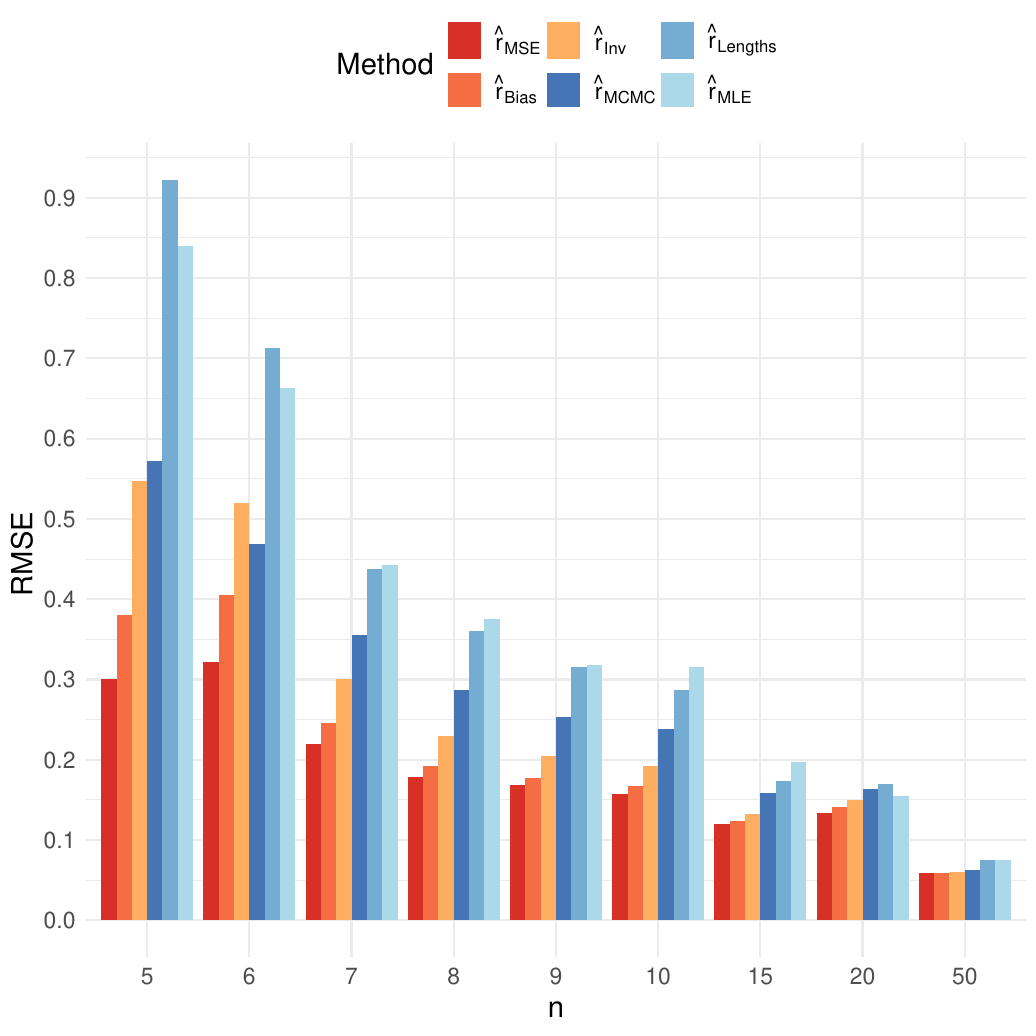}
        \caption{Comparison of the RMSE of the six estimators when $r = 0.5$ and $T = 40$.}
        \label{fig:mse1}
    \end{minipage}
    \hfill  
    \begin{minipage}{0.53\textwidth}
        \centering
        \includegraphics[width=\linewidth]{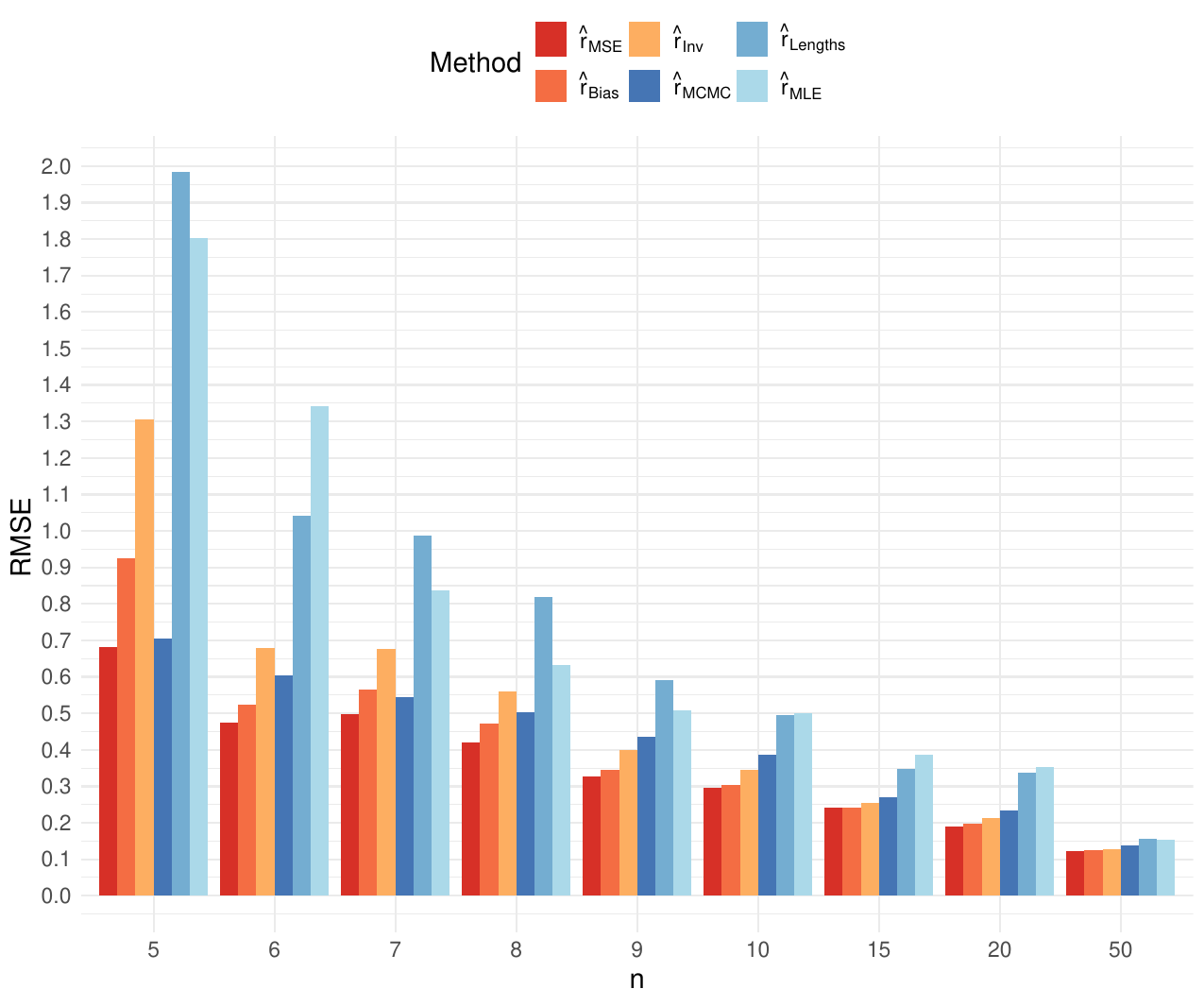}
        \caption{Comparison of the RMSE of the six estimators when $r = 1$ and $T = 40$.}
        \label{fig:mse2}
    \end{minipage}
\end{figure}

The mean absolute errors of the estimators are shown in Figures \ref{fig:MAD1} and \ref{fig:MAD2}. Here, $\hat r_{Bias}$ has the lowest mean absolute error across the six estimators. However, $\hat r_{MSE}$ and $\hat r_{Inv}$ also perform better than $\hat r_{MCMC}, \hat r_{MLE}$ and $\hat r_{Lengths}$ for most values of $n$, $r$, and $T$.

\begin{figure}[h!]
   \begin{minipage}{0.44\textwidth}
        \centering
        \includegraphics[width=\linewidth]{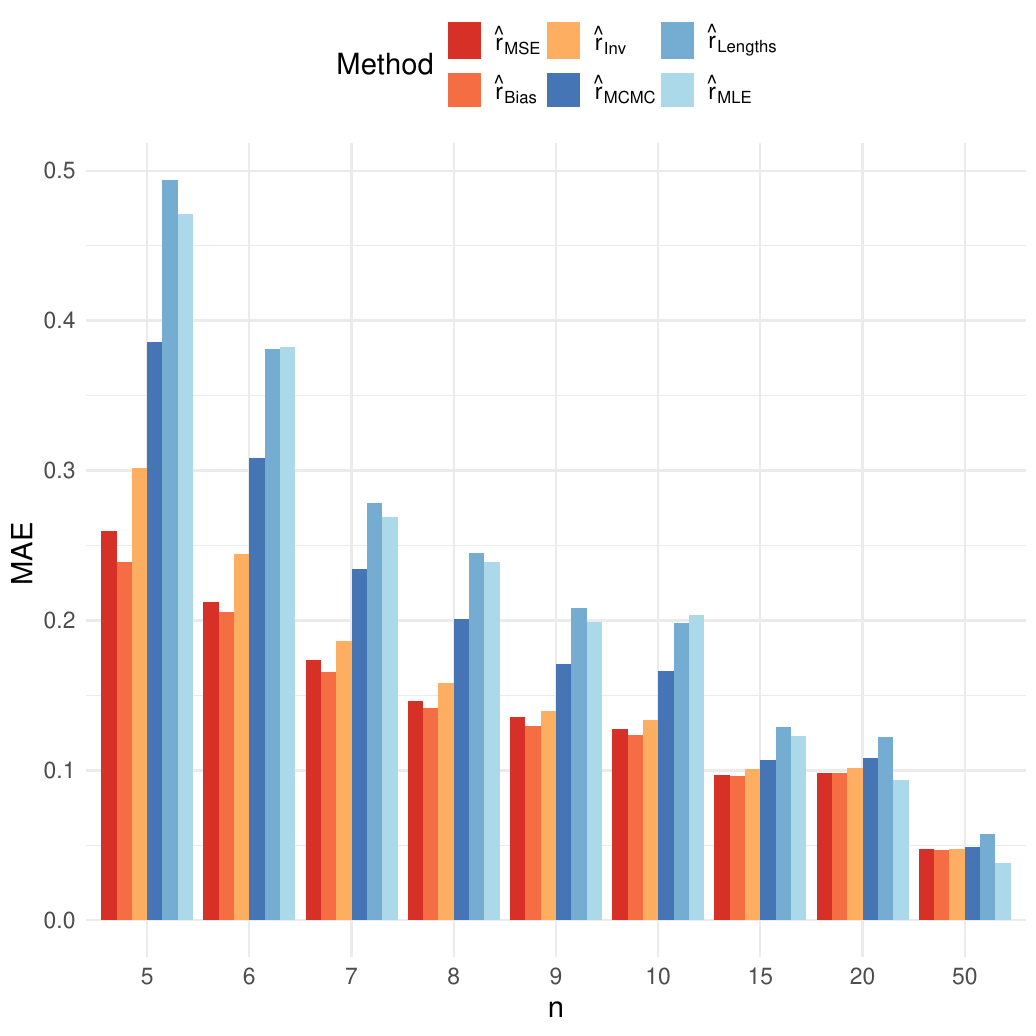}
        \caption{Comparison of the MAE of the six estimators when $r = 0.5$ and $T = 40$.}
        \label{fig:MAD1}
    \end{minipage}
    \hfill  
    \begin{minipage}{0.53\textwidth}
        \centering
        \includegraphics[width=\linewidth]{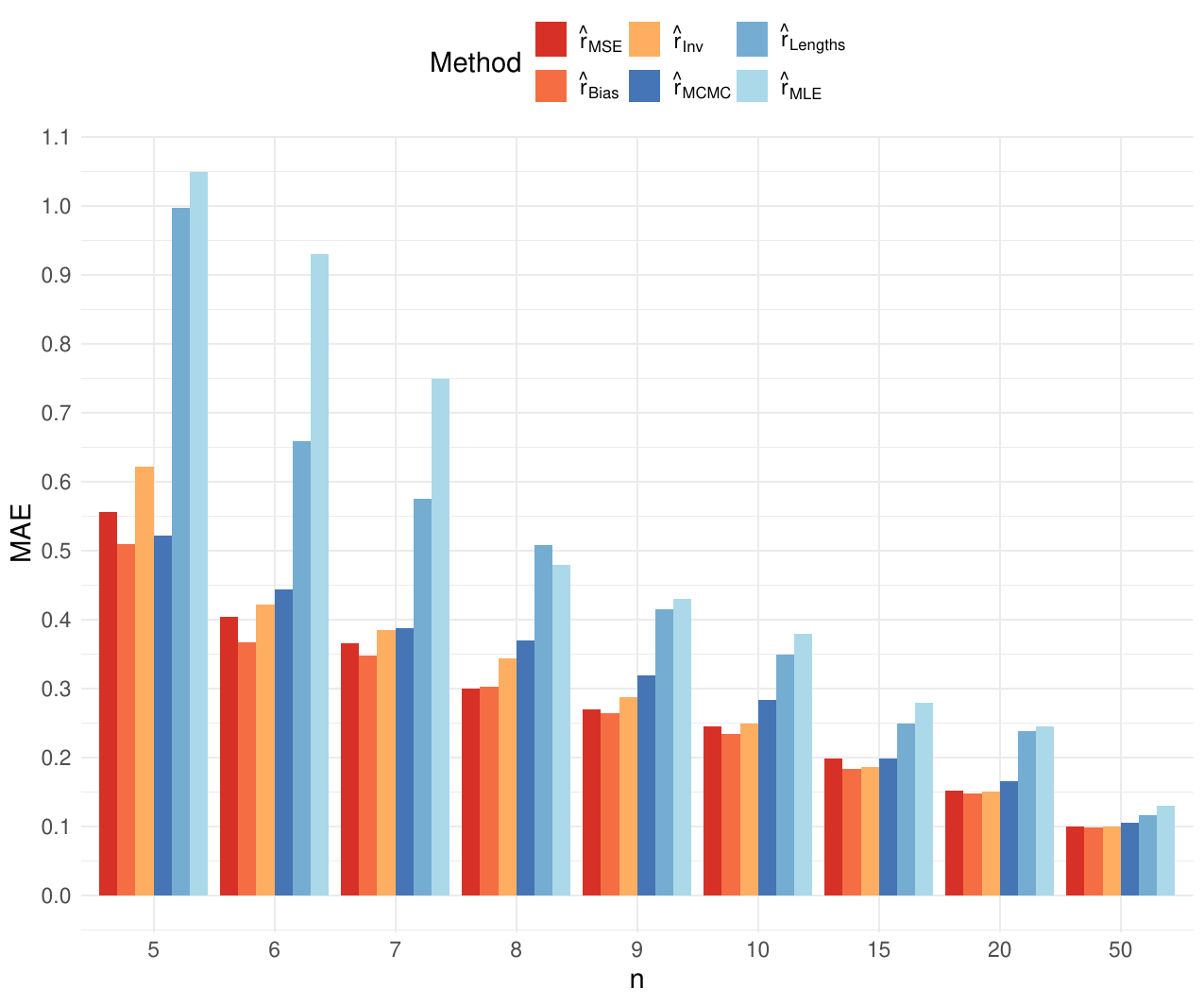}
        \caption{Comparison of the MAE of the six estimators when $r = 1$ and $T = 40$.}
        \label{fig:MAD2}
    \end{minipage}
\end{figure}

The bias of the estimators is presented in Figures \ref{fig:bias1} and \ref{fig:bias2}.  According to the figures, the estimates $\hat r_{Lengths}$, $\hat r_{MLE}$, and $\hat r_{MCMC}$ all tend to overestimate the growth rate on average for small $n$.  The estimate $\hat r_{Bias}$, which was designed to be exactly unbiased when $T = \infty$, remains nearly unbiased even when $T = 40$. Note that the estimate is nearly unbiased even though the mode of the density is considerably below the true value of $r$ because occasional large overestimates of $r$ affect the mean.
The estimate $\hat r_{Inv}$ has a slight upward bias, as would be expected from Jensen's inequality because it is the reciprocal of an estimator which is unbiased for $1/r$ when $T = \infty$.  The estimate $\hat r_{MSE}$ underestimates the growth rate on average when $n$ is small.  It acts as a shrinkage estimator, minimizing the MSE by introducing a downward bias to reduce the variance and mitigate the impact of rare events in which the growth rate is vastly overestimated.

\begin{figure}[h!]
   \begin{minipage}{0.44\textwidth}
        \centering
        \includegraphics[width=\linewidth]{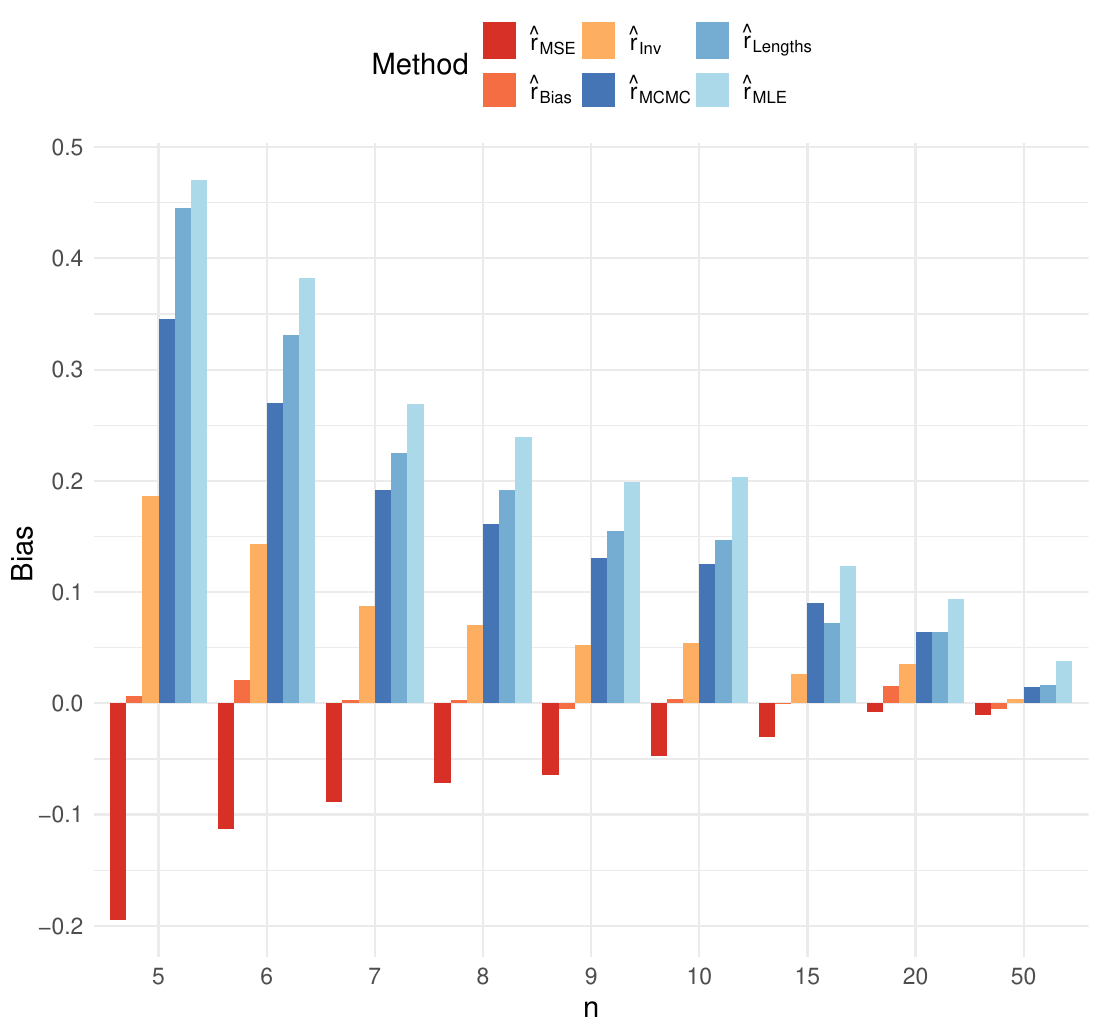}
        \caption{Comparison of the Bias of the six estimators when $r = 0.5$ and $T = 40$.}
\label{fig:bias1}
    \end{minipage}
    \hfill  
    \begin{minipage}{0.5\textwidth}
        \centering
        \includegraphics[width=\linewidth]{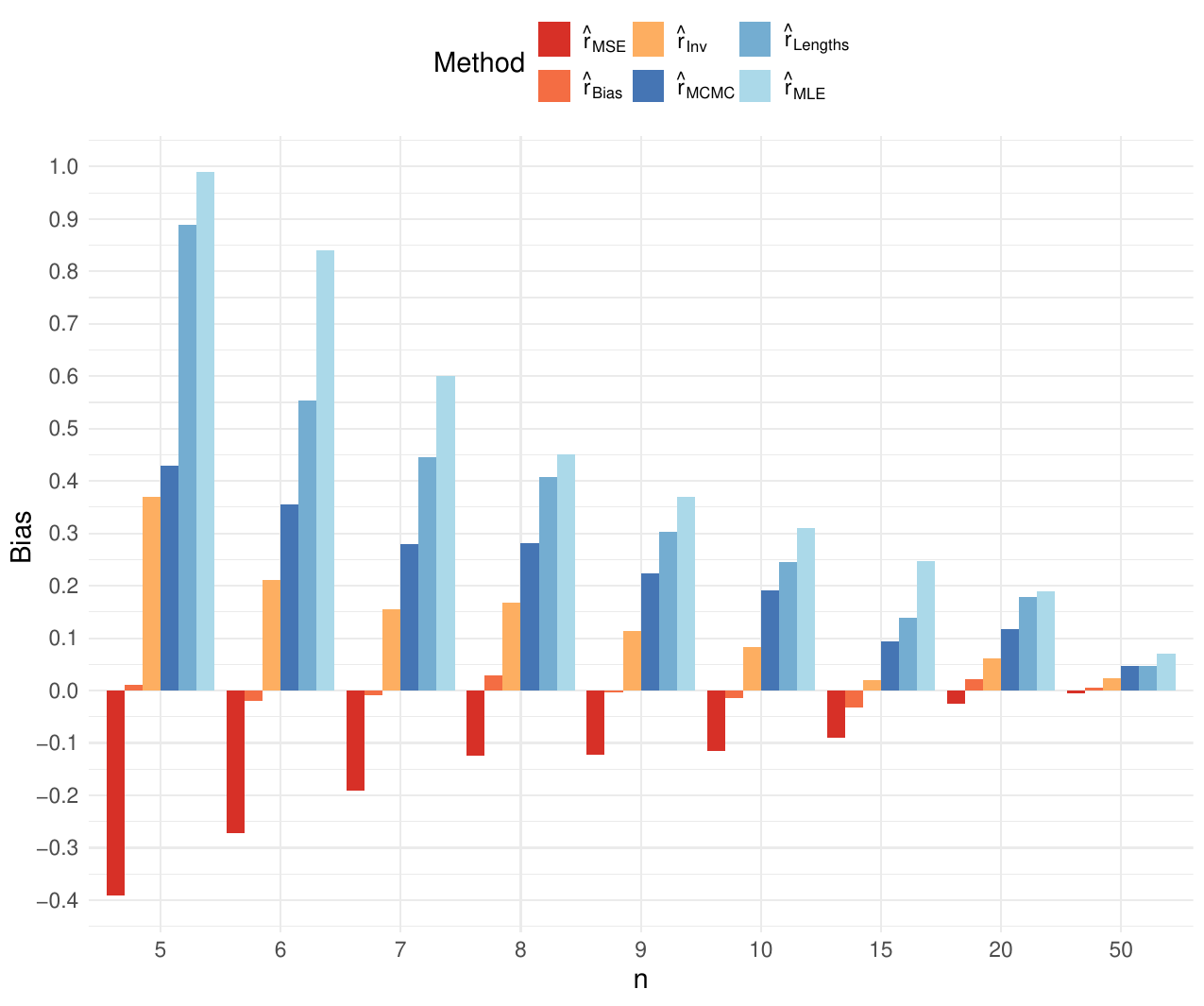}
        \caption{Comparison of the Bias of the six estimators when $r = 1$ and $T = 40$.}
        \label{fig:bias2}
    \end{minipage}
\end{figure}

We also compared the results when $n = 5$ and $T = 100$, and when $n = 20$ and $T = 100$, for several different values of $r$.  These results are shown in Figures~\ref{fig:T100_n5_mse}, \ref{fig:T100_n20_mse}, \ref{fig:T100_n_5_MAE}, and \ref{fig:T100_n_20_MAE}. These figures illustrate that the ratios $\mathrm{RMSE}/r$ and $\mathrm{MAE}/r$ for $\hat{r}_{{MSE}}$ and $\hat{r}_{{Inv}}$ remain roughly constant across the different values of $r$.  On the other hand, the ratios $\mathrm{RMSE}/r$ and $\mathrm{MAE}/r$
for $\hat{r}_{{MCMC}}$ tend to decrease as $r$ increases.  Consequently, by these criteria, the estimators $\hat{r}_{{MSE}}$ and $\hat{r}_{{Inv}}$ perform better for small $r$, while $\hat{r}_{{MCMC}}$ performs better for large $r$.

This behavior can be explained by the fact that the prior distribution imposes an upper bound of 4 on $\hat{r}_{{MCMC}}$, as this is the default value in the software package. As a result, the performance of $\hat{r}_{{MCMC}}$ is less impacted by rare events in which the other methods estimate extremely high values of $r$.
The impact of this is discussed in more detail in Section~\ref{sec:inf} in the appendix.  It is shown in Figure \ref{fig:mgr} in the appendix that the choice of the maximal growth rate substantially impacts the mean squared error of the estimator $\hat r_{MCMC}$. We believe that one of the benefits of our methods is that they do not require specifying a prior distribution in advance.

\begin{figure}[h!]
   \begin{minipage}{0.45\textwidth}
        \centering
        \includegraphics[width=\linewidth]{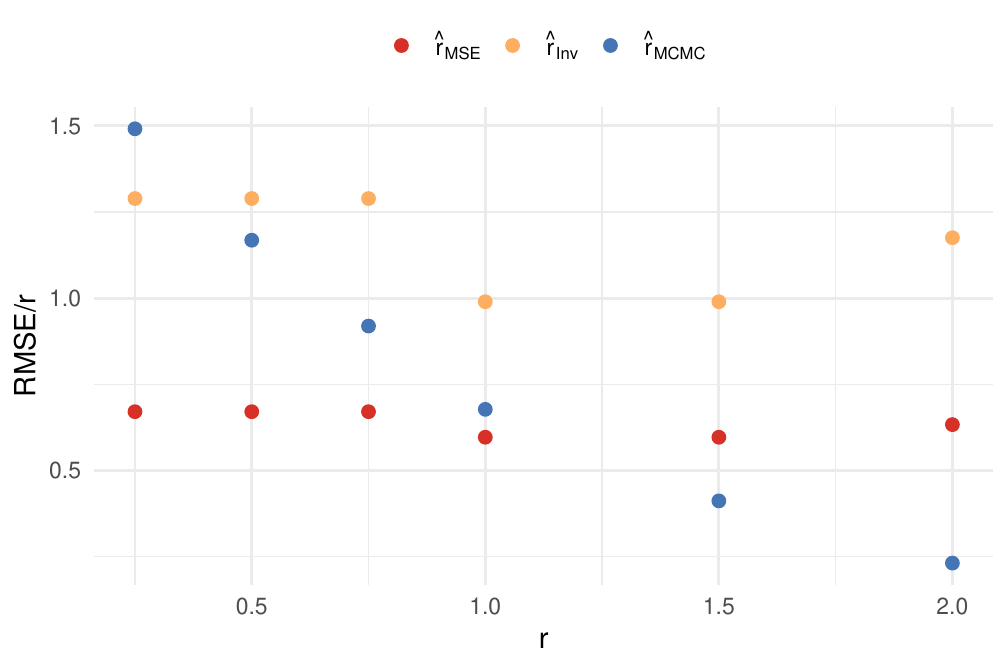}
        \caption{Comparison of the RMSE/r of three estimators when $n = 5$ and $T = 100$.}
        \label{fig:T100_n5_mse}
    \end{minipage}
    \hfill  
    \begin{minipage}{0.45\textwidth}
        \centering
        \includegraphics[width=\linewidth]{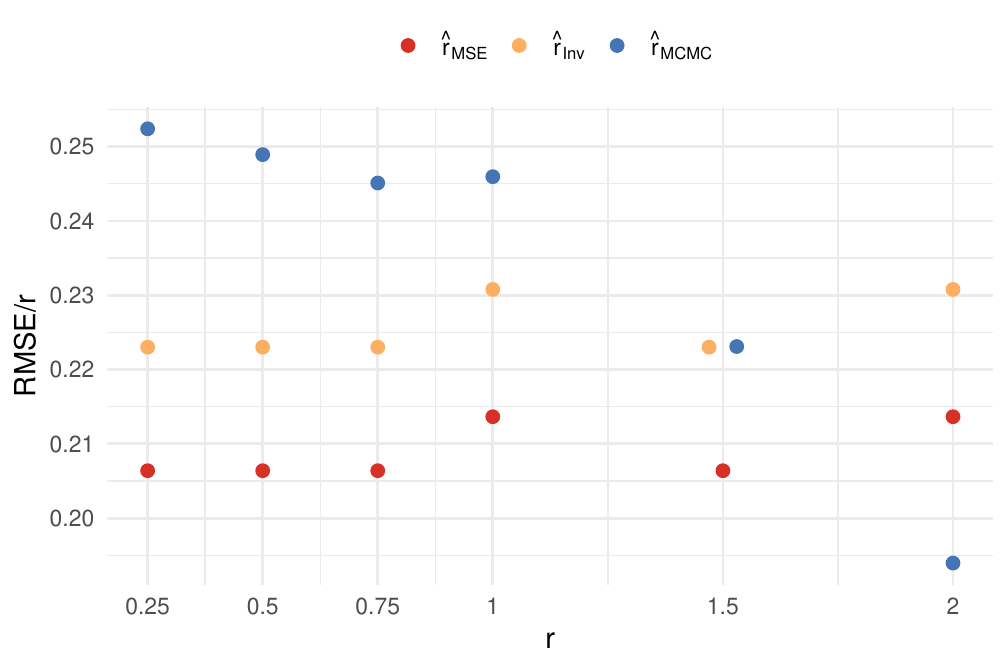}
        \caption{Comparison of the RMSE/r of three estimators when $n = 20$ and $T = 100$.}
        \label{fig:T100_n20_mse}
    \end{minipage}
\end{figure}
\begin{figure}[h!]
   \begin{minipage}{0.45\textwidth}
        \centering
        \includegraphics[width=\linewidth]{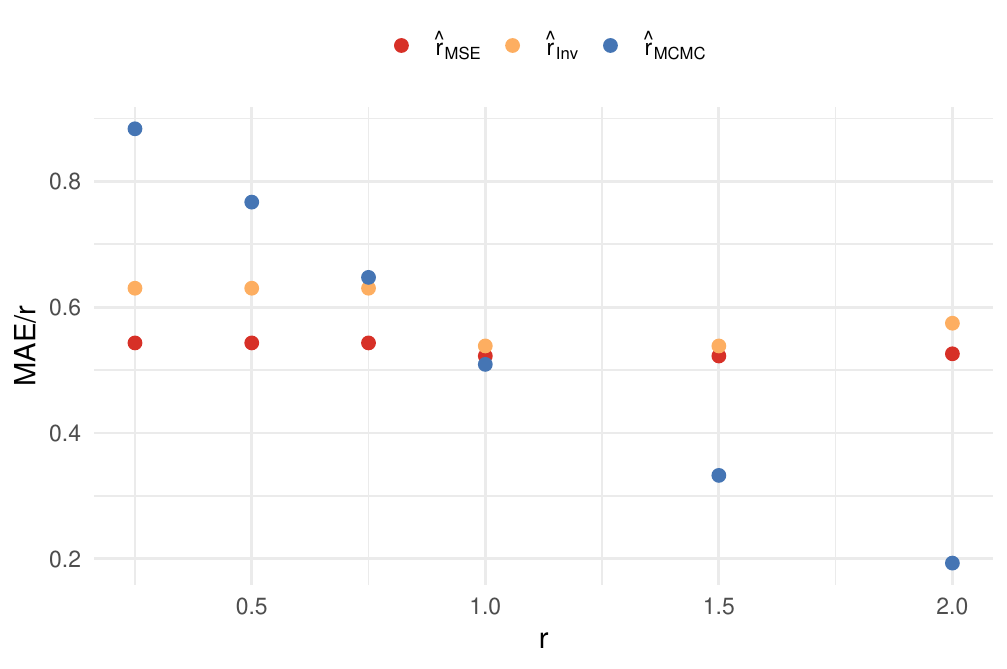}
        \caption{Comparison of the MAE/r of three estimators when $n = 5$ and $T = 100$.}
        \label{fig:T100_n_5_MAE}
    \end{minipage}
    \hfill  
    \begin{minipage}{0.45\textwidth}
        \centering
        \includegraphics[width=\linewidth]{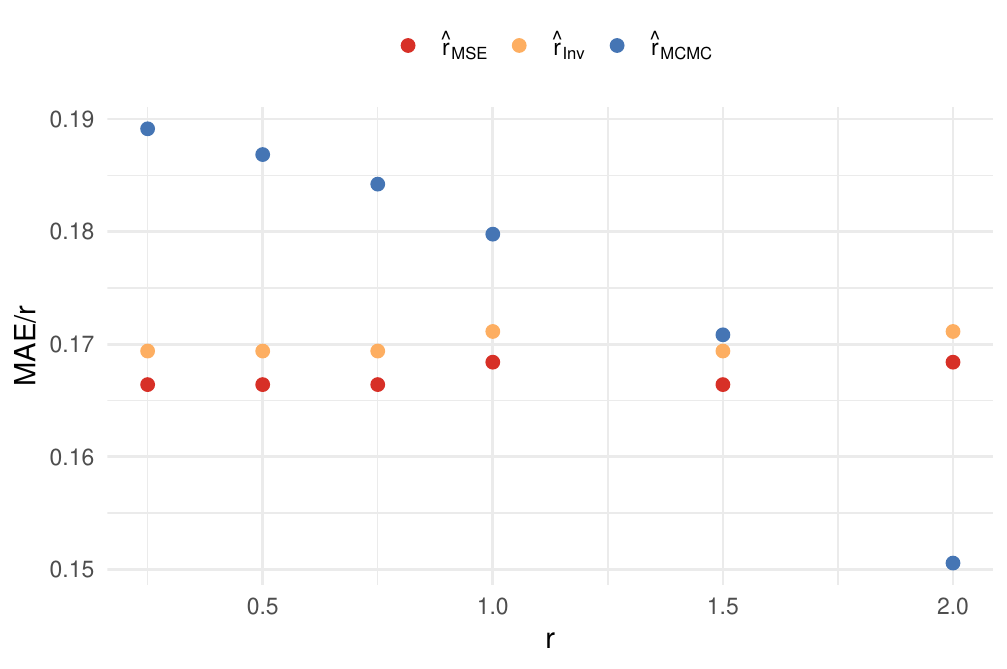}
        \caption{Comparison of the MAE/r of three estimators when $n = 20$ and $T = 100$.}
        \label{fig:T100_n_20_MAE}
    \end{minipage}
\end{figure}

We also compared different values of $T$. For example, when comparing Figures~\ref{fig:T100_n5_mse} and~\ref{fig:T100_n20_mse} with Figure~\ref{fig:mse2}, we observe relatively little difference in performance between $T = 100$ and $T = 40$. Overall, the estimators $\hat r_{MSE}$, $\hat r_{Inv}$, and $\hat r_{Bias}$ achieve good performance relative to other estimators.  They have the computational advantage that they do not require Markov Chain Monte Carlo, and they do not require the specification of a prior.  Which of the estimators is best will depend on the performance criteria.  The estimator $\hat r_{MSE}$ achieves the lowest mean squared error, while $\hat r_{Bias}$ has the advantage of being nearly unbiased.  On the other hand, both of these estimators tend to underestimate $r$ most of the time in small samples.  The estimator $\hat r_{Inv}$ appears to be a good compromise. It still achieves a lower mean absolute error than $\hat r_{MCMC}$ for small values of $r$, and the mode of the distribution is close to the true growth rate. It also has the benefit that $c_{Inv}(n)$ can be computed exactly.

\subsubsection{Performance of the confidence intervals}

We evaluate the confidence intervals as described in Section \ref{sec:CI}. First, we empirically approximate the quantiles \( q_{0.025} \) and \( q_{0.975} \).  The values of these quantiles for certain values of $n$ are listed in the Appendix.  Next, we calculate $\hat r$ using $c(n) = 1$ for 1000 simulated coalescent point processes, denoted by \( \hat{r}_i \) for \( i = 1, \dots, 1000 \), and for \( n \in \{5, 6, 7, 8, 9, 10, 15, 20\} \), $T = 40$, and $r = 1$. Using these values along with the quantiles \( q_{0.025} \) and \( q_{0.975} \), we compute the confidence intervals for each sample \( i \) as \( \textup{CI}_i := [\hat{r}_i / q_{0.025}, \hat{r}_i / q_{0.975}] \), for \( i = 1, \dots, 1000 \). To evaluate the performance of these intervals, we calculate the coverage probability, that is, the percentage of the simulations for which \( r \in \textup{CI}_i \).

The empirical coverage probabilities are shown in Figure \ref{fig:cov1}.  For all values of $n$, the coverage probability is close to 0.95.
The empirical quantiles were chosen to ensure that the confidence intervals give exactly 95$\%$ coverage in the limit as $T \rightarrow \infty$, so these simulation results demonstrate that the confidence intervals remain quite accurate for $T = 40$.

\begin{figure}[h!]
\centering
    \includegraphics[width=0.5\linewidth]{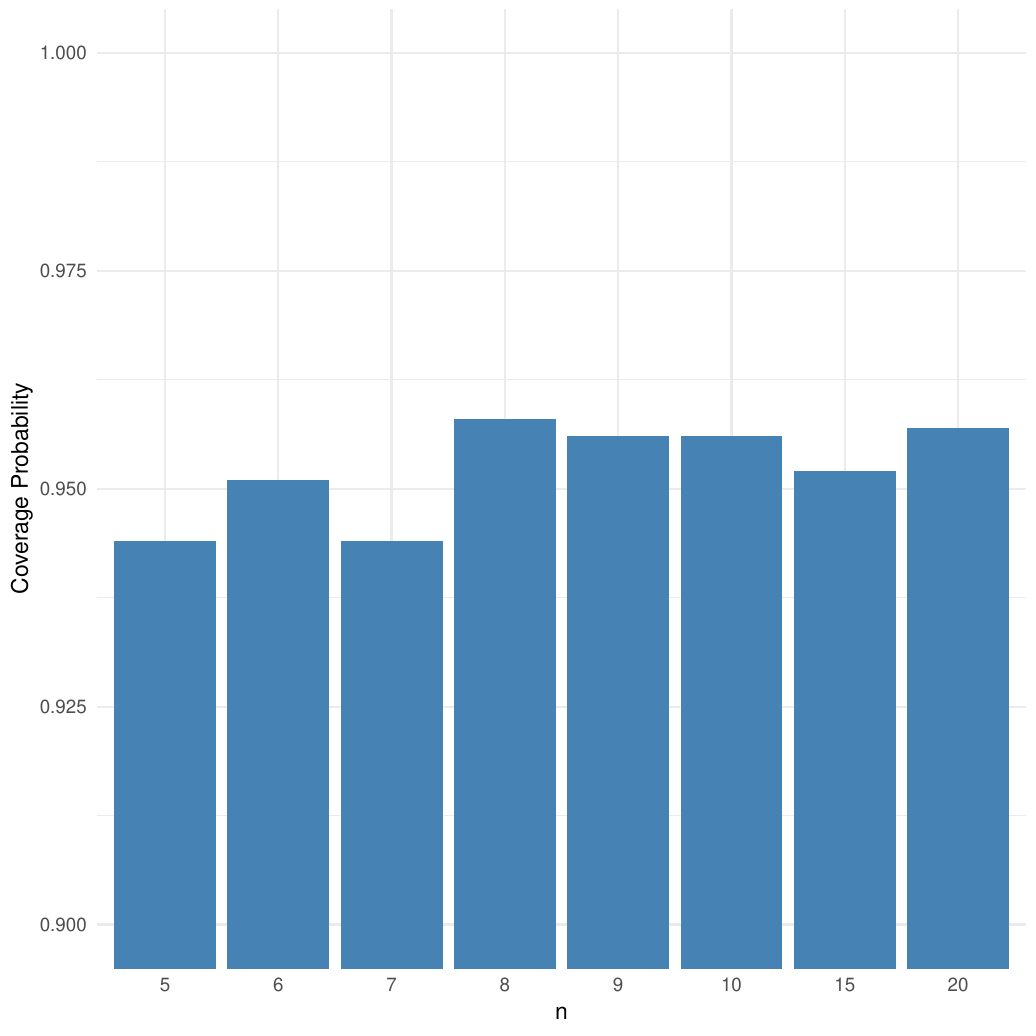}
    \caption{Coverage probability of the confidence intervals for $r = 1$ and $T = 40.$}\label{fig:cov1}
\end{figure}

\subsubsection{Influence of the Approximation  \texorpdfstring{$T = \infty$}{T = infinity} on the Quality of the Estimator}

For the calculations of \( c_{MSE}(n) \), \( c_{Bias}(n) \), and $c_{Inv}(n)$ we used the approximation \( T \to \infty \) instead of considering a fixed value of $T$.  It was necessary to make this choice because otherwise, the optimal constant would depend on the unknown value of \( r \).

Now, we quantify the error introduced by this approximation. Given $n$, $T$, and $r$, we can determine by simulation, for any value of the constant $c$, the root mean squared error and the absolute bias of the resulting estimator of $r$.  In Figures \ref{fig:influence_mse} and \ref{fig:influence_bias}, we plot the RMSE and the bias as a function of the constant.

\begin{figure}[h!]
\begin{minipage}{0.45\textwidth}
    \centering
    \includegraphics[width=\linewidth]{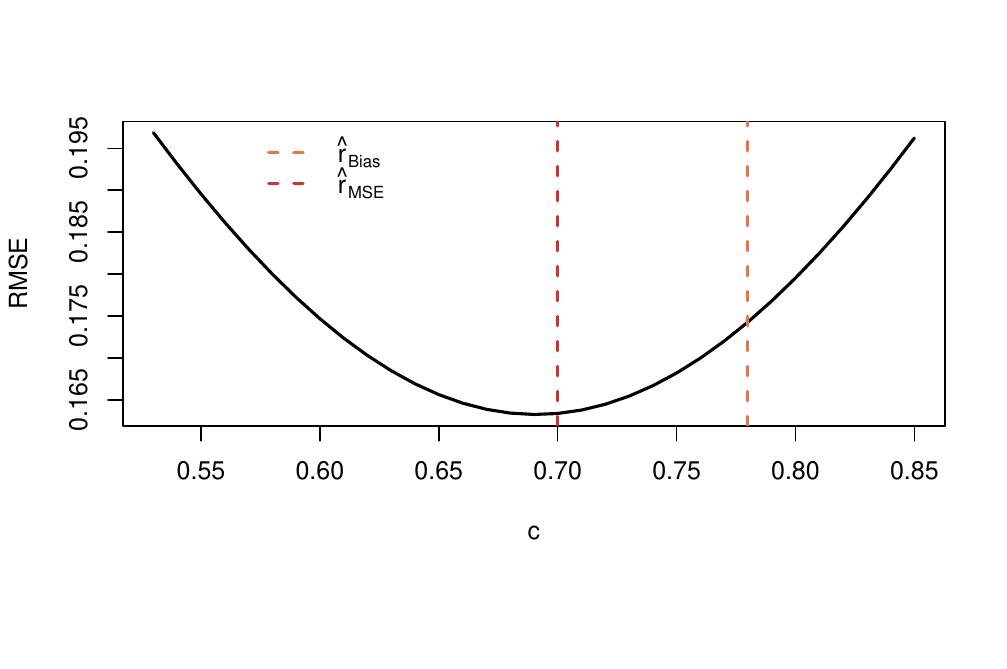}
    \caption{Influence of the constant on the RMSE for $T = 40$, $n = 10$, and $r = 0.5$.}
    \label{fig:influence_mse}
    \end{minipage}
    \hfill  
    \begin{minipage}{0.45\textwidth}
    \centering
    \includegraphics[width=\linewidth]{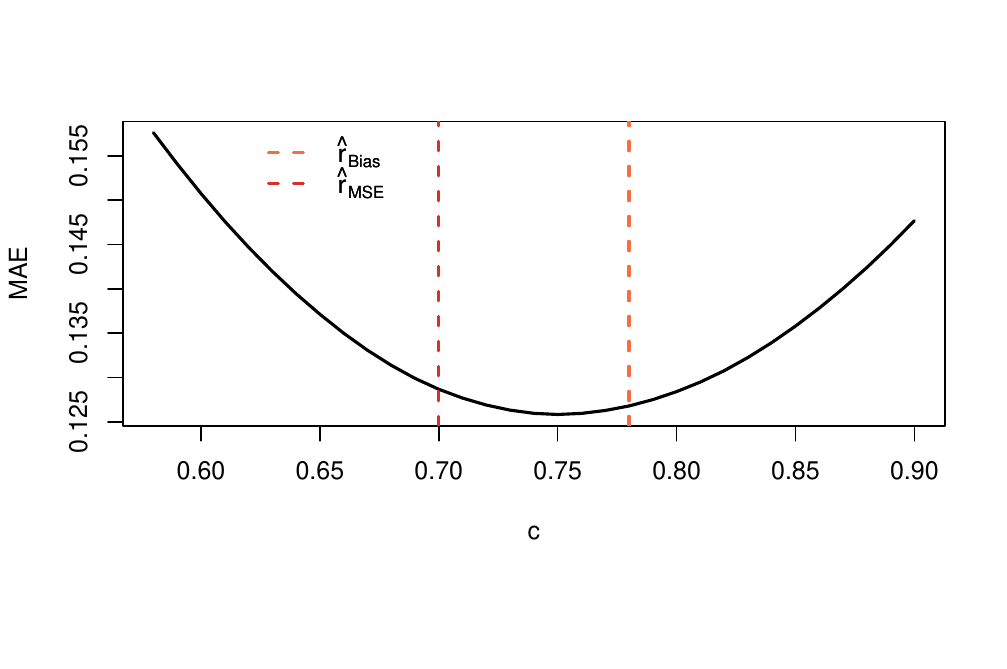}
    \caption{Influence of the constant on the MAE for $T = 40$, $n = 10$, and $r = 0.5$.}
    \label{fig:influence_bias}
    \end{minipage}
\end{figure}

The red dotted line in Figure~\ref{fig:influence_mse} lies close to the minimum point of the curve, which means that the difference between the constant that minimizes the MSE for $T=40$ and $r=1$, and the one that minimizes it for $T =\infty$, is small. Similarly, Figure \ref{fig:influence_bias} shows that the difference between the constant minimizing the absolute bias for $T = 40$ and $r = 1$ and the one that minimizes the absolute bias for $T = \infty$ is also small. Thus, both figures indicate that the error introduced by assuming $T = \infty$ when calculating the constants is negligible. 

\subsection{Theoretical properties of the estimators}

Here we show that, even in the limit as $n \rightarrow \infty$, our estimate achieves better performance than the estimate based on internal branch lengths, $\hat r_{Lengths},$ considered in \cite{Johnson2023}.  

When $n$ is large, the coalescence times are well approximated by the formula in \eqref{eq:H_i}, and $c(n)$ is close to $1$.  Therefore, we will consider the estimator $$\hat r = \frac{(n-1)(n-2)}{\sum_{i=1}^{n-1} \sum_{j=1}^{n-1} (H_i - H_j)^+}.$$  The following lemma gives the asymptotic behavior of $\hat r$.  Here we denote by $\mathcal N(\mu, \sigma^2)$ the normal distribution with mean $\mu$ and variance $\sigma^2$.

\begin{lemma}\label{rhatlem}
We have the convergence in distribution $$\sqrt{n}(\hat r - r) \xRightarrow{n \to \infty} \mathcal N \bigg(0, r^2 \bigg(4 - \frac{\pi^2}{3} \bigg) \bigg).$$
\end{lemma}

\begin{proof}
We will first compute the mean and variance of the reciprocal of the estimator.
Recall that $(H_i - H_j)^+ = \frac{1}{r}(U_i - U_j)^+$.
Integral calculations in \cite{Johnson2023} using \eqref{Udensity} give
\begin{align*}
\mathbb E\big[(U_1 - U_2)^+\big] &= 1 \\
\mathbb E\big[((U_1 - U_2)^+)^2\big] &= \frac{\pi^2}{3} \\
\mathbb E\big[(U_1 - U_2)^+(U_2 - U_3)^+\big] &= 2 - \frac{\pi^2}{6}.
\end{align*}
A similar integral calculation gives
\begin{align}\label{eq:ev_U2}
  \mathbb E\big[(U_1 - U_2)^+(U_1 - U_3)^+\big] = 2.
\end{align}

We have
$$\mathbb E \bigg[ \frac{1}{\hat r} \bigg] = \frac{1}{(n-1)(n-2)} \sum_{i=1}^{n-1} \sum_{j=1}^{n-1} \mathbb E \big[ (H_i - H_j)^+ \big] = \frac{1}{r}.$$
Also,
\begin{align}
\mbox{Var} \bigg( \frac{1}{\hat r} \bigg) &=\mbox{Var}\bigg(\frac{1}{(n-1)(n-2)} \sum_{i=1}^{n-1} \sum_{j=1}^{n-1} (H_i - H_j)^+ \bigg) \nonumber \\
&= \frac{1}{(n-1)^2 (n-2)^2 r^2} \sum_{i=1}^{n-1} \sum_{j=1}^{n-1} \sum_{k=1}^{n-1} \sum_{\ell = 1}^{n-1} \mbox{Cov}\Big( (U_i - U_j)^+, (U_k - U_{\ell})^+ \Big). \label{4sumcov}
\end{align}
The covariance in \eqref{4sumcov} is zero if $i = j$, if $k = \ell$, or if $i,j,k,\ell$ are all distinct.  There are $(n-1)(n-2)$ terms with $(i,j) = (k,\ell)$.  These terms have covariance $\frac{\pi^2}{3} - 1$.  There are $2(n-1)(n-2)(n-3)$ terms in which either $i = k$ or $j = \ell$ but the other two indices are different.  Using that $U_i$ and $-U_i$ have the same distribution, these terms have covariance $2 - 1 = 1$.  There are also $2(n-1)(n-2)(n-3)$ terms in which either $i = \ell$ or $j = k$ but the other two indices are different.  These terms have covariance $2 - \frac{\pi^2}{6} - 1 = 1 - \frac{\pi^2}{6}$.  It follows that
$$\mbox{Var} \bigg( \frac{1}{\hat r} \bigg) = \frac{(n-1)(n-2)(\frac{\pi^2}{3} - 1) + 2(n-1)(n-2)(n-3)(2 - \frac{\pi^2}{6})}{(n-1)^2 (n-2)^2 r^2} \sim \frac{1}{nr^2} \bigg(4 - \frac{\pi^2}{3} \bigg),$$ where $\sim$ means that the ratio of the two sides tends to one as $n \rightarrow \infty$.

The asymptotic normality of $1/\hat r$ can be deduced from the fact that $1/\hat r$ is a $U$-statistic.  Indeed, Theorem 7.1 in \cite{hoeffding1992class} states that the distribution of random variables $V$ of the form
$$V := \frac{1}{n(n-1)} \sum_{i=1}^n \sum_{\substack{j=1 \\ j \neq i}}^n f(V_i, V_j)$$
is asymptotically normal, provided that $f$ is symmetric in its arguments, that the first and second moments of $f(V_1, V_2)$ exist, and that $V_1, \dots, V_n$ are i.i.d. These conditions are fulfilled for $1/\hat r$, so we can directly deduce the asymptotic normality
\begin{align*}
   \sqrt{n}\left(\frac{1}{\hat r} - \frac{1}{r}\right) \xRightarrow{n \to \infty} \mathcal N\left(0, \frac{1}{r^2}\left(4- \frac{\pi^2}{3}\right)\right)
\end{align*}
from the variance computation above and Theorem 7.1 in \cite{hoeffding1992class}.
Applying the delta method now yields the conclusion of the lemma.
\end{proof}

It remains to show that the estimator $\hat r_{Inv}$ has the same asymptotic behavior as $\hat r$.  We will assume that we have a sample of $n$ individuals taken from a birth and death process at time $T_n$.  Here $T_n$ is finite, so the coalescence times are given by \eqref{HinT}.  We need to assume that $n \rightarrow \infty$, and that $T_n \rightarrow \infty$ sufficiently fast relative to $n$.  More precisely, we assume, as in \cite{Johnson2023}, that
\begin{equation}\label{Tnassump}
\lim_{n \rightarrow \infty} n^{3/2} (\log n) e^{-r T_n} = 0.
\end{equation}

\begin{theorem}
Suppose \eqref{Tnassump} holds.  Then
\begin{align}
    \sqrt{n}(\hat r_{Inv}- r)\xRightarrow{n \to \infty} \mathcal N\left(0, r^2\left(4- \frac{\pi^2}{3}\right)\right). \label{eq:normal}
\end{align} 
\end{theorem}

\begin{proof}
By Lemma \ref{rhatlem}, it is enough to show that $\hat r_{Inv}$ and $\hat r$ can be coupled in such a way that $\sqrt{n}(\hat r_{Inv} - \hat r) \xrightarrow[]{n \to \infty}_p 0$.  We have
\begin{align*}
&\sqrt{n}(\hat r_{Inv} - \hat r) \\
&\quad = \sqrt{n} \bigg( \frac{c_{Inv}(n)(n-1)(n-2)}{\sum_{i=1}^{n-1} \sum_{j=1}^{n-1} (H_{i,n,T_n} - H_{j,n,T_n})^+} - \frac{c_{Inv}(n)(n-1)(n-2)}{\sum_{i=1}^{n-1} \sum_{j=1}^{n-1} (H_i - H_j)^+}\bigg) + \sqrt{n}(c_{Inv}(n) - 1) \hat r.
\end{align*}
We note that $\sqrt{n} (c_{Inv}(n) - 1) \xrightarrow[]{n \to \infty} 0$ according to \eqref{eq:inv}, so the second term converges in probability to zero.  It remains to show that the first term converges in probability to zero.

Using that $(H_i - H_j)^+ - (H_{i,n,T_n} - H_{j,n,T_n})^+ \leq |H_i - H_{i,n,T_n}| + |H_j - H_{j,n,T_n}|,$ the first term is bounded above by
\begin{align}\label{Hdiff}
&n^{5/2} c_{Inv}(n) \left(\frac{\left( \sum_{i=1}^{n-1} \sum_{j=1}^{n-1} (H_i - H_j)^+ \right)  - \left( \sum_{i=1}^{n-1} \sum_{j=1}^{n-1} (H_{i,n,T_n} - H_{j,n,T_n})^+ \right) }{\left(\sum_{i=1}^{n-1} \sum_{j=1}^{n-1} (H_{i,n,T_n} - H_{j,n,T_n})^+\right)\left(\sum_{i=1}^{n-1} \sum_{j=1}^{n-1} (H_i - H_j)^+ \right) }\right) \nonumber \\
&\quad \leq n^{5/2} c_{Inv}(n) \left(\frac{\sum_{i=1}^{n-1} \sum_{j=1}^{n-1} \big( |H_i - H_{i,n,T_n}| + |H_j - H_{j,n,T_n}| \big)}{\left(\sum_{i=1}^{n-1} \sum_{j=1}^{n-1} (H_{i,n,T_n} - H_{j,n,T_n})^+\right)\left(\sum_{i=1}^{n-1} \sum_{j=1}^{n-1} (H_i - H_j)^+ \right) }\right).
\end{align}
By Lemmas 4 and 5 in the supplemental material to \cite{Johnson2023}, the random variables can be coupled so that for a sequence of events $(A_n)_{n=1}^{\infty}$ with $\lim_{n \rightarrow \infty} \mathbb P(A_n) = 1$, we have $$\lim_{n \rightarrow \infty} \sqrt{n} \, \mathbb E\big[|H_1 - H_{1,n,T_n}| {\bf 1}_{A_n}\big] = 0,$$
from which it follows that
$$n^{-3/2} \sum_{i=1}^{n-1} \sum_{j=1}^{n-1} \big( |H_i - H_{i,n,T_n}| + |H_j - H_{j,n,T_n}| \big) \xrightarrow[]{n \to \infty}_p 0.$$
Because the denominator in \eqref{Hdiff} is on the scale of $n^4$ and $\lim_{n \rightarrow \infty} c_{Inv}(n) = 1$, it follows that the expression in \eqref{Hdiff} converges to zero in probability as $n \rightarrow \infty$.
\end{proof}

It was shown in \cite{Johnson2023} that $\sqrt{n}(1/\hat r_{Lengths} - 1/r) \xRightarrow{n \to \infty} \mathcal N(0, 1/r^2).$  Applying the delta method yields
\begin{align*}
\sqrt{n}(\hat r_{Lengths} - r)  &\xRightarrow{n \to \infty} \mathcal N\left(0, r^2\right).
\end{align*}
Because $4 - \frac{\pi^2}{3} \approx 0.71$, we see that the asymptotic distribution of $\hat r_{Inv}$ has a lower variance than that of $\hat r_{Lengths}$.  Because maximum likelihood estimates are asymptotically efficient, we would not expect $\hat r_{Inv}$ to have better asymptotic performance for large $n$ than $\hat r_{MLE}$.

\subsection{Application to Real Data}

We apply $\hat r_{MSE}$, $\hat r_{Inv}$, and $\hat r_{MCMC}$ to data from \cite{williams2022life, van2021reconstructing, fabre2022longitudinal, mitchell2022clonal} as saved, e.g., in \texttt{realCloneData[["cloneTrees"]]} of the package \texttt{cloneRate.}
These were the same 42 data sets analyzed in \cite{Johnson2023}.
For each of these data sets, the sample size $n$ was between $10$ and $109$. The estimates of the growth rate are shown in Figure \ref{fig:real}.

Usually, it holds that $\hat r_{MSE} \leq \hat r_{Inv} \leq \hat r _{MCMC}$ for the considered data. For most of the data sets, the three estimators give very similar results, although there are a few exceptions where there are larger differences among the methods.
As would be expected from the simulation study,
the differences among the estimators decreases for increasing sample size $n.$ Overall, we see that, much like the methods proposed by \cite{Johnson2023}, the methods $\hat r_{Inv}$ and $\hat r_{MSE}$ give results on these real data sets that are comparable to those obtained from Markov chain Monte Carlo methods, for less computational cost.

\begin{figure}[h!]
    \centering
    \includegraphics[width=0.95\linewidth]{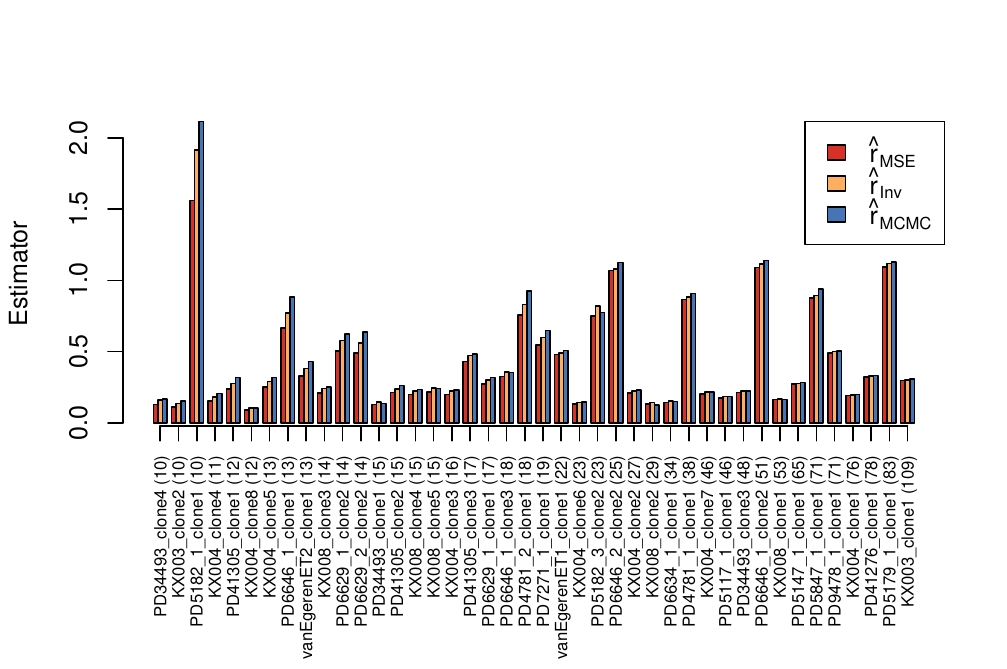}
    \caption{Application to data from \cite{williams2022life, van2021reconstructing, fabre2022longitudinal, mitchell2022clonal}. On the $x$-axis are the names of the clone and the sample size and on the $y$-axis are the estimators. }\label{fig:real}
\end{figure}

\section{Discussion}

We provide an estimator for the growth rate of a birth and death processes which outperforms the state-of-the-art methods (\cite{williams2022life, mitchell2022clonal, Johnson2023}) with respect to the mean squared error and the bias, especially for small sample sizes. We also show analytically that for large sample sizes, our estimator outperforms an estimate based on internal branch lengths from \cite{Johnson2023}.

Our estimator involves a constant which can be chosen in multiple ways.  We propose three choices for the constant, called $c_{MSE}(n)$, $c_{Bias}(n)$, and $c_{Inv}(n)$.  The estimate with $c_{MSE}(n)$ will minimize the mean squared error but will underestimate the growth rate most of the time in small samples.  The estimate using $c_{Bias}(n)$ is nearly unbiased but has a slightly higher mean squared error, and its density also has a mode that is substantially below the true growth rate when $n$ is small.  The estimate using $c_{Inv}(n)$ has a slightly higher mean squared error than the other two estimates, but its mode is close to the true growth rate, and the value of the constant can be computed analytically. 

Our method only applies to exponentially growing populations. However, tumors do not always follow an exponential growth curve (\cite{zahid2021forecasting}).  Nevertheless, as noted in \cite{Johnson2023}, as long as the population grows exponentially for a short time, our methods should provide a good estimate of that initial exponential growth rate because the coalescence times of the sampled lineages are usually early in the history of the population.  Mathematical results in this direction were established by \cite{cheek2022coalescent}.  To test the assumption of an exponentially growing population, one could use skyline methods as introduced by \cite{pybus2000integrated}, \cite{drummond2005}, and \cite{stadlerskyline}. Skyline methods can reconstruct the population size history under different demographic models.

Our methods also rely on the assumption that all individuals have the same birth and death rates.  This is a limitation because in cell populations, sometimes several driver mutations accumulate, leading to fitness differences among the cells (\cite{scott2017somatic}).  However, it was noted by \cite{Johnson2023} that site frequency spectrum data from the blood cancer subclones studied in \cite{Johnson2023} agrees well with theoretical predictions for an exponentially growing population in which all cells have the same fitness.  Finally, our methods do not take into account any spatial structure in the population and therefore may not be directly applicable to solid tumors.

\bigskip
\bigskip
\noindent {\bf {\Large{Data and Code Availability}}}
\bigskip

\noindent The code is available on our \href{https://github.com/CarolaHeinzel/Estimation_NetGrowthRate}{GitHub website} (\url{https://github.com/CarolaHeinzel/Estimation_NetGrowthRate}). This code can be used to reproduce the results, i.e. to compare the different estimators for the growth rate and to apply our method to new data. The code is based on the R-package \texttt{cloneRate}, where the authors also include the real data that has been used in our study.

\bigskip
\bigskip
\noindent {\bf {\Large{Funding}}}
\bigskip

\noindent CSH is funded by the Deutsche Forschungsgemeinschaft (DFG, German Research Foundation) – Project-ID 499552394 – SFB Small Data. CSH received travel support from the Wissenschaftliche Gesellschaft in Freiburg im Breisgau.

\bigskip
\bigskip
\noindent {\bf {\Large{Acknowledgments}}}
\bigskip

\noindent CSH thanks Peter Pfaffelhuber for his continuous and strong support.  Both authors thank Kit Curtius and Brian Johnson for helpful comments on an earlier draft of the paper, which led to an improved presentation. Both authors thank two referees for comments that improved the presentation of the results.

\bigskip
\bigskip
\noindent {\bf {\Large{Declaration of Conflicts of Interest}}}
\bigskip

\noindent The authors declare that there are no conflicts of interest.

\printbibliography

\section{Appendix}\label{sec:appendix}

\subsection{Derivation of equation \eqref{eq:ev}}\label{app:der}

In this section, we calculate $\mathbb E\big[L_n^{in}|H_n^{(1)}, ..., H_n^{(n-1)}\big].$  Recall from \eqref{eq:IBL} that
$$E\big[L_n^{in}|H_n^{(1)}, ..., H_n^{(n-1)}\big] = \mathbb E\bigg[\bigg( \max_{1 \leq i \leq n-1} H_{i,n} - H_{1,n} \bigg) + \sum_{i=1}^{n-2} (H_{i,n} - H_{i+1,n})^+ \Big| H_n^{(1)}, ..., H_n^{(n-1)}\bigg].$$
Recall that $H_{i,n} = T - \frac{1}{r}(\log(Q_n) + U_{i,n})$, and the random variables $U_{1,n}, \dots, U_{n-1,n}$ are i.i.d.  The random variables 
$H_{1,n}, \dots, H_{n-1, n}$ are not i.i.d. because they all depend on $Q_n$, but these random variables are exchangeable.  Therefore, it holds that $\mathbb P(H_{i,n} = H^{(j)}_n) = \frac{1}{n-1}$ for all $i, j \in \{1,..., n-1\}$.  It follows that
$$\mathbb E\big[H_{i,n}|H_n^{(1)}, ..., H_n^{(n-1)}\big] = \frac{1}{n-1} \sum_{i=1}^{n-1} H_n^{(i)} = \frac{1}{n-1} \sum_{i=1}^{n-1} H_{i,n}.$$ We also calculate
\begin{align*}
    \mathbb E\bigg[\sum_{i=1}^{n-2} (H_{i,n}-H_{i+1,n})^+|H_n^{(1)}, \dots, H_n^{(n-1)}\bigg]
    &= \sum_{i=1}^{n-2} \frac{1}{(n-1)(n-2)} \sum_{v=1}^{n-1} \sum_{u=1}^{n-1} (H_n^{(u)}-H_n^{(v)})^+ \\
    &= \frac{1}{n-1} \sum_{v=1}^{n-1} \sum_{u=1}^{n-1} (H_n^{(u)}-H_n^{(v)})^+ \\
     &= \frac{1}{n-1} \sum_{i=1}^n \sum_{j=1}^n (H_{i,n} - H_{j,n})^+.
\end{align*}
Combining these results, we get
$$\mathbb E\big[L_n^{in}|H_n^{(1)}, \dots, H_n^{(n-1)}\big] = H_n^{(n-1)} - \frac{1}{n-1}\sum_{i=1}^{n-1} H_{i,n} + \frac{1}{n-1}\sum_{i = 1}^{n-1} \sum_{j = 1}^{n-1}(H_{i,n} - H_{j,n})^+,$$
which matches \eqref{eq:ev}.

 \subsection{Calculation of  \texorpdfstring{$c_{Inv}(n)$}{c(n)}}\label{sec:cInv}

We show here how to obtain the value of $c_{Inv}(n)$.  We have
$$c_{Inv}(n) = \mathbb E \bigg[ \frac{1}{S_n} \bigg] = \mathbb E[(U_{1,n} - U_{2,n})^+].$$
Therefore,
\begin{align*}
 c_{Inv}(n) &=  \int_{-\infty}^\infty \mathbb P(U_{i,n} \leq x \leq U_{j,n}) \: dx \\
  &=  \int_{0}^\infty \int_{-\infty}^\infty \mathbb P(U_{i,n} \leq x \leq U_{j,n} |Q_{n} = q) f_{Q_{n}(q)}\: dx \: dq \\
   &=  \int_{0}^\infty \int_{-\infty}^\infty  \mathbb P(U_{i,n} \leq x|Q_{n} = q) \mathbb P(x \leq U_{j,n} |Q_{n} = q) f_{Q_{n}(q)} \: dx \: dq
\end{align*}
Recalling \eqref{Undensity}, we have, for $x \geq -\log q$,
$$\mathbb P(U_{j,n} \geq x|Q_n = q) = \int_x^{\infty} \frac{1+q}{q} \cdot \frac{e^u}{(1+e^u)^2} \: du = \frac{1+q}{q} \cdot \frac{1}{1+e^x}.$$
Therefore,
$$c_{Inv}(n) = \int_{0}^\infty \int_{-\log q}^\infty \left(1- \frac{1+q}{q}\frac{1}{1+e^x}\right) \bigg( \frac{1+q}{q} \frac{1}{1+e^x} \bigg) \frac{nq^{n-1}}{(1+q)^{n+1}} \: dx \: dq.$$
Now making the substitution $y = 1/(1+e^x)$, so that $dy/dx = -y(1-y)$, we get
$$c_{Inv}(n) = \int_{0}^\infty \bigg( \int_{0}^{q/(1+q)} \left(1- \frac{1+q}{q}y\right) \frac{1}{1-y} \: dy \bigg) \frac{nq^{n-2}}{(1+q)^{n}} \: dq.$$
To evaluate the inner integral, we write $1/(1-y)$ as an infinite series to get
\begin{align*}
\int_{0}^{q/(1+q)} \left(1- \frac{1+q}{q}y\right) \frac{1}{1-y} \: dy &= \int_{0}^{q/(1+q)} \bigg(\sum_{j=0}^{\infty} y^j - \frac{1+q}{q} \sum_{j=0}^{\infty} y^{j+1} \bigg) \: dy \\
&= \sum_{j=0}^{\infty} \bigg( \frac{q}{1+q} \bigg)^{j+1} \frac{1}{j+1} - \sum_{j=0}^{\infty} \bigg( \frac{q}{1+q} \bigg)^{j+1} \frac{1}{j+2} \\
&= \sum_{j=0}^{\infty} \bigg(\frac{1}{j+1} - \frac{1}{j+2}\bigg) \bigg( \frac{q}{1+q} \bigg)^{j+1}.
\end{align*}
It follows that
\begin{align*}
c_{Inv}(n) &= \sum_{j=0}^\infty \left(\frac{1}{j+1} - \frac{1}{j+2}\right) \int_0^{\infty} \frac{nq^{n+j-1}}{(1+q)^{n+j+1}} \: dq \\
&= \sum_{j=0}^\infty \left(\frac{1}{j+1} - \frac{1}{j+2}\right) \frac{n}{n+j} \\
&= \frac{n}{n-1} \sum_{j=0}^{\infty} \bigg( \frac{1}{1+j} - \frac{1}{n+j} \bigg) - \frac{n}{n-2} \sum_{j=0}^{\infty} \bigg( \frac{1}{2+j} - \frac{1}{n+j} \bigg) \\
&= \frac{n}{n-1} \sum_{k=1}^{n-1} \frac{1}{k} - \frac{n}{n-2} \sum_{k=2}^{n-1} \frac{1}{k} \\
&= \frac{n}{n-2} \bigg(1 - \frac{1}{n-1} \sum_{k=1}^{n-1} \frac{1}{k} \bigg).
\end{align*}

\subsection{Extended Information about the real data}\label{sec:EI}

We show the distribution of $T$ and $n$ in the data from \cite{williams2022life}, \cite{van2021reconstructing, fabre2022longitudinal, mitchell2022clonal} as available in the R package \texttt{cloneRate}.

\begin{figure}[h!]
\begin{minipage}{0.45\textwidth}
    \centering
    \includegraphics[width=\linewidth]{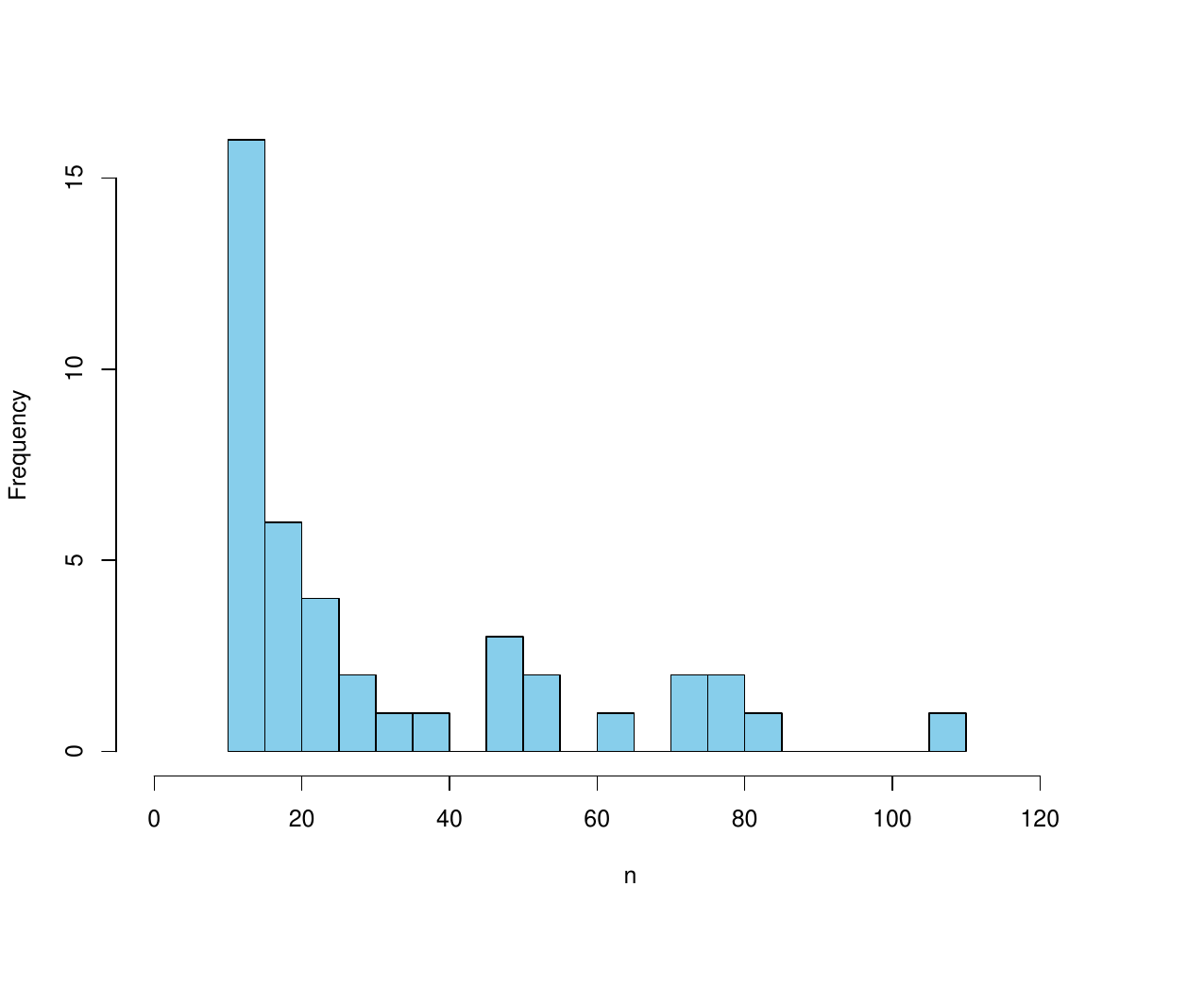}
    \caption{Distribution of $n$ for the real data from \cite{williams2022life}, \cite{van2021reconstructing, fabre2022longitudinal, mitchell2022clonal}.}
    \end{minipage}
    \hfill  
    \begin{minipage}{0.53\textwidth}
    \centering
    \includegraphics[width=\linewidth]{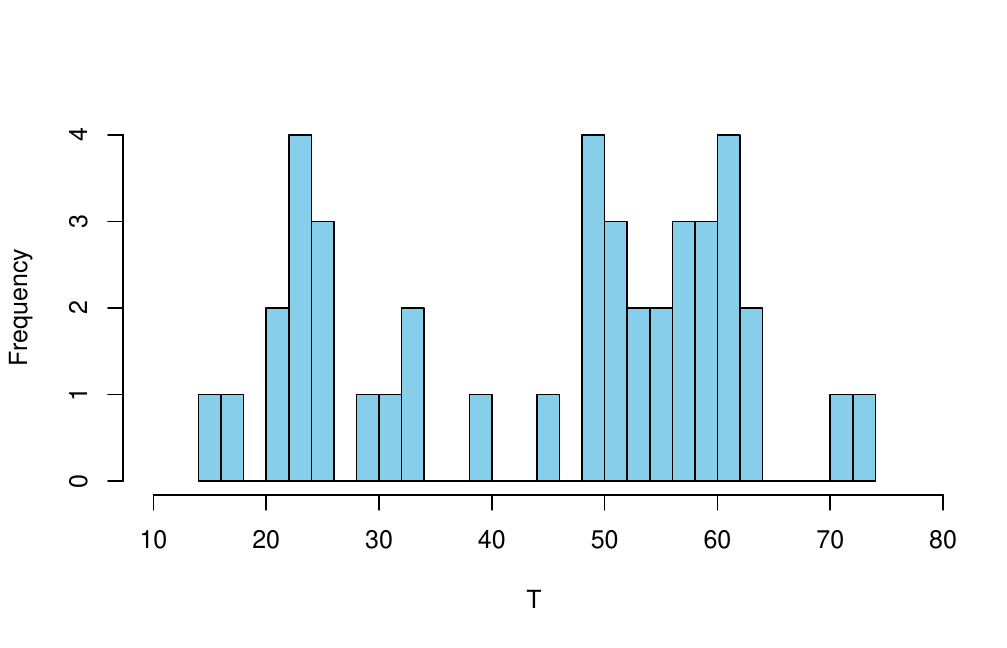}
    \caption{Distribution of $T$ for the real data from \cite{williams2022life}, \cite{van2021reconstructing, fabre2022longitudinal, mitchell2022clonal}.}
    \end{minipage}
\end{figure}

\subsection{Values of the constants}

Table \ref{tab:c} displays the values of the constants $c_{Inv}(n)$$, c_{MSE}(n)$, and $c_{Bias}(n)$ to two decimal places for certain values of $n$.  The table also gives the reciprocals of the 0.025 and 0.975 quantiles of the distribution of $S_n$, which are used in calculating the confidence intervals. The code to calculate these quantiles can be found on our \href{https://github.com/CarolaHeinzel/Estimation_NetGrowthRate}{GitHub website}.

\begin{table}[h!]
    \centering
    \begin{tabular}{c|c|c|c|c|c}
       $n$  & $c_{Inv}(n)$ &$c_{MSE}(n)$ & $c_{Bias}(n)$ & $1/q_{0.025}$ & $1/q_{0.975}$ \\
         \hline
       5  & 0.80 & 0.355 & 0.59 & 1.73 & 0.21 \\
6  & 0.82 & 0.49  & 0.66 & 1.62 & 0.27 \\
7  & 0.83 & 0.58  & 0.71 & 1.53 & 0.32 \\
8  & 0.84 & 0.63  & 0.74 & 1.49 & 0.36 \\
9  & 0.85 & 0.67  & 0.76 & 1.46 & 0.40 \\
10 & 0.86 & 0.70  & 0.78 & 1.43 & 0.44 \\ 
11 & 0.86 & 0.72  & 0.80  & 1.40 & 0.45\\
12 & 0.87 & 0.74  & 0.81  & 1.38 & 0.47\\
13 & 0.88 & 0.76  & 0.82  & 1.36 & 0.49\\
14 & 0.88 & 0.77  & 0.83  & 1.35 & 0.51 \\
15 & 0.89 & 0.79  & 0.84 & 1.33 & 0.52 \\
16 & 0.89 & 0.80  & 0.84  & 1.32 & 0.53 \\
17 & 0.89 & 0.81  & 0.85  & 1.31 & 0.55\\
18 & 0.90 & 0.82  & 0.86  & 1.30 & 0.56 \\
19 & 0.90 & 0.82  & 0.86 & 1.29 & 0.57 \\
20 & 0.90 & 0.83  & 0.87  & 1.28 & 0.58 \\
30 & 0.93& 0.87 & 0.92&   1.22& 0.65 \\
40 & 0.94 & 0.90 & 0.93 & 1.20 & 0.71 \\
50 & 0.95 & 0.92  & 0.93& 1.18  & 0.74  \\
60 & 0.95 & 0.93&  0.95 & 1.16 & 0.76\\
70 & 0.96 & 0.94& 0.96& 1.15 &  0.78\\
80 & 0.96 & 0.94 & 0.96 & 1.14 &  0.79\\
90 & 0.96 & 0.95 & 0.96 & 1.13 &  0.80\\
100 & 0.97 & 0.95 & 0.96 & 1.13 &  0.81 \\
    \end{tabular}
    \caption{Values of the constants.}
    \label{tab:c}
\end{table}

\subsection{Influence of the maximal growth rate on the estimator}\label{sec:inf}

In Figure \ref{fig:mgr}, we present the RMSE for $n=5, 10, 20$ and $r = 1$ with different values of the maximal growth rate ($4, 5, 10$), using the R-function \texttt{birthDeathMCMC} of the R package \texttt{cloneRate}. There, we see that performance improves, especially for small sample sizes, drastically when we reduce the maximal growth rate from 10 to 5. We performed 100 simulation runs for each case.

\begin{figure}[h!]
    \centering
    \includegraphics[width=0.5\linewidth]{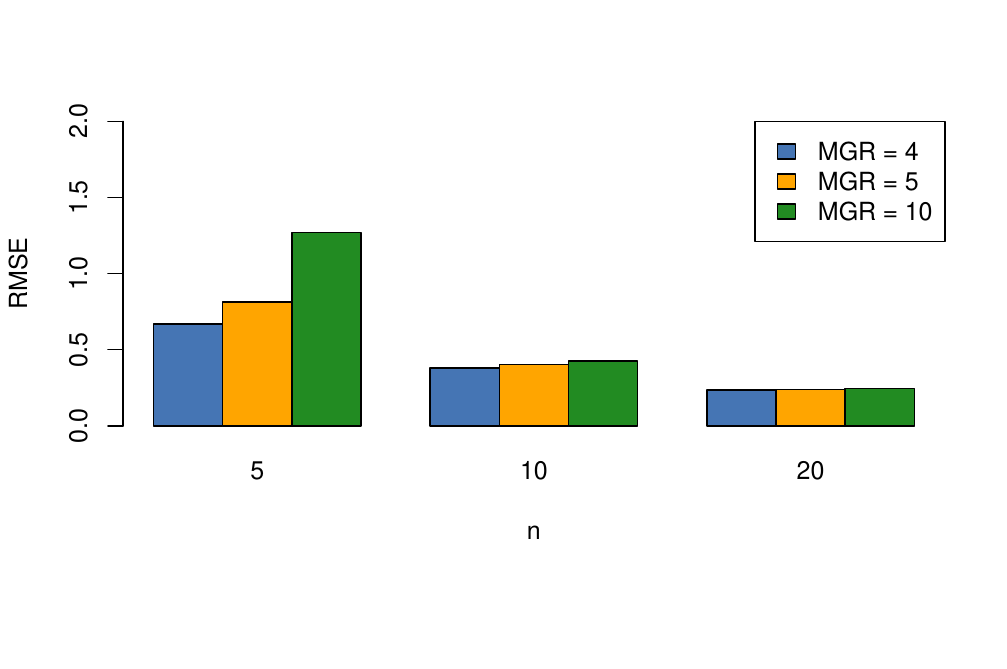}
    \caption{Influence of the value for \texttt{maxGrowthRate} in the function \texttt{birthDeathMCMC} in the R-package \texttt{cloneRate}. Here, MGR stands for maximal growth rate.}
    \label{fig:mgr}
\end{figure}

\end{document}